\documentclass[11pt,final]{article}
\usepackage{fullpage}

\usepackage{amsmath,amssymb,amsthm}
\usepackage{booktabs} 
\usepackage{xspace}
\usepackage{accents} 
\usepackage{xparse} 
\usepackage{threeparttable}

\usepackage[colorlinks=true]{hyperref}
\hypersetup{
	linkcolor=[rgb]{0.4,0,0.6},
	citecolor=[rgb]{0, 0.4, 0},
	urlcolor=[rgb]{0.6, 0, 0}
}

\usepackage{cleveref} 
\usepackage{array}
\usepackage{booktabs}
\usepackage{multicol}
\usepackage{tabularx}
\usepackage{framed}
\usepackage{thmtools, thm-restate}
 \usepackage{mathtools}
 \usepackage[square,longnamesfirst]{natbib}
\usepackage{verbatim}

\usepackage{color}

\newcommand{\mfnote}[1]{{\color{blue}{#1}}}

\newcommand{\aknote}[1]{{\color{purple}{#1}}}

\usepackage{todonotes}

\usepackage{accsupp}

\usepackage[shortlabels]{enumitem}

\newcommand{\argmax}{\operatorname{arg\,max}}

\newcommand{\prob}[2][]{\text{\bf Pr}\ifthenelse{\not\equal{}{#1}}{_{#1}}{}\!\left[#2\right]}
\newcommand{\expect}[2][]{\text{\bf E}\ifthenelse{\not\equal{}{#1}}{_{#1}}{}\!\left[#2\right]}

\def\Pr{\ensuremath{\mathrm{Pr}}}

\def\argmax{\ensuremath{\mathrm{argmax}}}

\crefformat{equation}{(#2#1#3)} 

\newcommand{\rep}{{\tilde{s}}}

\newcommand{\Reals}{\mathbb{R}}
\newcommand{\Realsp}{\Reals^+}

\newenvironment{proofof}[1]{\begin{proof}[Proof of #1]}{\end{proof}}

\newcommand{\E}{\mathbb{E}}

\newtheorem{theorem}{Theorem}[section]
\newtheorem{lemma}[theorem]{Lemma}
\newtheorem{remark}[theorem]{Remark}

\newtheorem{proposition}[theorem]{Proposition}

\newtheorem{observation}[theorem]{Observation}

\newtheorem{definition}{Definition}[section]
\newtheorem{example}{Example}[section]

\theoremstyle{definition}

\newcommand*{\pr}[2][]{\text{Pr}\ifx\\\left[#1\right]\\\else_{#1}\fi \left[#2\right]}

\newcommand{\bzero}{\mathbf {0}}

\newcommand{\bs}{\mathbf s}

\begin{document}
\title{Combinatorial Auctions with Interdependent Valuations: \\SOS to the Rescue
\thanks{The work of A. Eden and M. Feldman was partially supported by the European Research Council under the
	European Unions Seventh Framework Programme (FP7/2007-2013) / ERC grant agreement number 337122, by
	the Israel Science Foundation (grant number 317/17), and by an Amazon research award. The work of A. Eden and A. Fiat was partially supported by	ISF 1841/14. The work of A. Karlin and K. Goldner was supported by NSF grants CCF-1420381 and CCF-1813135. The work of K.	Goldner was also supported by a Microsoft Research PhD Fellowship.
}
}
\author{Alon Eden%
\thanks{%
    {Tel Aviv University (\url{alonarden@gmail.com})}}
\and Michal Feldman%
\thanks{%
    {Tel Aviv University (\url{michal.feldman@cs.tau.ac.il})}}
\and Amos Fiat%
\thanks{%
    {Tel Aviv University (\url{fiat@tau.ac.il})}}
\and Kira Goldner%
\thanks{%
    {University of Washington (\url{kgoldner@cs.washington.edu})}}
\and Anna R. Karlin%
\thanks{%
	{University of Washington (\url{karlin@cs.washington.edu})}}
}

\maketitle

\begin{abstract}
We study combinatorial auctions with interdependent valuations.
In such settings, each agent $i$ has a private signal $s_i$ that captures her private information, and the valuation function of every agent depends on the entire signal profile, ${\bs}=(s_1,\ldots,s_n)$.
The literature in economics shows that the interdependent model gives rise to strong impossibility results, and identifies assumptions under which optimal solutions can be attained.
The computer science literature provides approximation results for simple single-parameter settings (mostly single item auctions, or matroid feasibility constraints).
Both bodies of literature focus largely on valuations satisfying a technical condition termed {\em single crossing} (or variants thereof).

We consider the class of {\em submodular over signals} (SOS) valuations (without imposing any single-crossing type assumption), and provide the first welfare approximation guarantees for multi-dimensional combinatorial auctions, achieved by universally ex-post IC-IR mechanisms.
Our main results are:
$(i)$ 4-approximation for any single-parameter downward-closed setting with single-dimensional signals and SOS valuations;
$(ii)$ 4-approximation for any combinatorial auction with multi-dimensional signals and {\em separable}-SOS valuations; and
$(iii)$ $(k+3)$- and $(2\log(k)+4)$-approximation for any combinatorial auction with single-dimensional signals, with $k$-sized signal space, for SOS and strong-SOS valuations, respectively. All of our results extend to a parameterized version of SOS, $d$-SOS, while losing a factor that depends on $d$.

\end{abstract}

\section{Introduction}

%
%

Maximizing social welfare with private valuations is a solved problem.
The classical Vickrey-Clarke-Grove (VCG) family of mechanisms~~\citep{vickrey1961counterspeculation,clarke,groves}, of which the Vickrey second-price auction is a special case, are dominant strategy  incentive-compatible and guarantee optimal social welfare in general social choice settings.

In this paper, we consider combinatorial auctions, where each agent has a value for every subset of items, 
and the goal is to maximize the social welfare, namely the sum of agent valuations for their assigned bundles.
As a special case of general social choice settings, the VCG mechanism solves this problem optimally, \emph{as long as the values are independent}.

There are many settings, however, in which the independence of values is not realistic.  If  the item being sold has money-making potential or is likely to be resold,
the values different agents have may be correlated, or perhaps even common.
A classic example is an auction for the right to drill for oil in a certain location~~\citep{wilson1969communications}. Importantly, in such settings, agents may have different information about what that value actually is. For example, the value of an oil lease depends on how much oil there actually is, and the different agents may have access to different assessments about this. Consequently, an agent might change her own estimate of the value of the oil lease given access to the information another agent has. Similarly, if an agent had access to the results of a house inspection performed by a different agent, that might change her own estimate of the value of a house that is for sale.

The following model due to \cite{milgrom1982theory}, described here for single-item auctions, has become standard for  auction design in such settings. These are known as {\em interdependent value settings} (IDV)~\footnote{ See also~~\citep{krishna2009auction,milgrom2004putting}.} and are defined as follows:

\begin{itemize}
\item Each agent $i$ has a real-valued, private {\em signal} $s_i$.
The set of signals $\bs = (s_1, s_2, \ldots, s_n)$ may be
drawn from a (possibly) correlated distribution.

The signals summarize the information available to the agents about the item. For example, when the item to be sold is a house, the signal could capture the results of an inspection and privately collected information about the school district. In the setting of oil drilling rights, the signals could be information that each companies' engineers have about the site based on geologic surveys, etc.

\item The {\em value of the item} to agent $i$ is
a function $v_i(\bs)$ of the signals (or information) of {\em all} agents.

A typical example
is when $v_i(\bs) = s_i + \beta \sum_{j\ne i}s_j$, for some $\beta \le 1$. This
type of valuation function captures settings where an agent's value depends both on how much he likes the item ($s_i$) and on the
resale value which is naturally estimated in terms of how much other agents like the item ($\sum_{j\ne i} s_j$)~~\citep{myerson1981optimal}.
\end{itemize}

In the economics literature, interdependent settings have been studied for about 50 years (with far too many papers to list; for an overview, see~~\citep{krishna2009auction}).
Within the theoretical computer science community,  interdependent (and correlated) settings have received less attention (see Section \ref{sec:related} for further discussion and references).

\subsection{Maximizing Social Welfare }

Consider the goal of maximizing social welfare in interdependent settings. Here, a direct revelation mechanism consists of each agent $i$ reporting a bid for their private signal $s_i$, and the auctioneer determining the allocation and payments. (It is assumed that the auctioneer knows the form of the valuation functions $v_i (\cdot)$.)


In interdependent settings, it is not possible\footnote{ Except perhaps in degenerate situations.} to design dominant-strategy incentive-compatible auctions,
since an agent's value depends on {\em all} of the signals, so if, say, agent $i$ misreports his signal, then agent $j$ might win at a price above her value if she reports truthfully. The next strongest equilibrium notion one could hope for is to maximize efficiency in ex-post equilibrium:  bidding truthfully is an {\em ex-post equilibrium} if an agent does not regret having bid truthfully, given that other agents bid truthfully.
In other words, bidding truthfully is a Nash equilibrium for every signal profile.\footnote{ Note that, of course, every ex-post equilibrium is a Bayes-Nash incentive compatible equilibrium, but not necessarily vice versa, and therefore ex-post equilibria are much more robust: they do not depend on knowledge of the priors and bidders need not think about how other bidders might be bidding.  This increases our confidence that an ex-post equilibrium is likely to be reached.}
A strong impossibility result due to \cite{jehiel2001efficient} shows that with {\em multi-dimensional} signals, maximizing welfare is generically impossible even in Bayes-Nash equilibrium.\footnote{For more details on this and other related work, see Section~\ref{sec:related}.}

For single-item auctions with single-dimensional signals, a characterization of ex-post incentive compatibility in the IDV setting is known, analogous to Myerson's characterization for the independent private values model (e.g.,~\citet{RTCoptimalrev}). The characterization says that there are payments that yield an  ex-post incentive-compatible mechanism if and only if the corresponding allocation rule is monotone in each agent's signal, when all other signals are held fixed.
 Maximizing efficiency in ex-post equilibrium is also provably impossible unless
the valuation functions $v_i(\bs)$ satisfy a technical condition known as the  {\em single-crossing condition}~~\citep{milgrom1982theory,Aspremont82,maskin1992,ausubel1999generalized,dasgupta2000efficient,athey01,bergemann2009information,CFK,che2015efficient,li2016approximation,RTCoptimalrev}.
I.e., the influence of agent $i$'s signal on his own value is at least as high as its influence on other agents' values, when
all other signals $\bs_{-i}$ are held fixed~\footnote{This implies that given signals $\bs_{-i}$,
if agent $i$ has the highest value when $s_i = s^*$, then agent $i$ continues to have the highest
value for $s_i > s^*$. This is precisely the monotonicity needed for ex-post incentive compatibility.}. When the single-crossing condition holds, there is a generalization of VCG that maximizes efficiency in ex-post equilibrium. (See ~\citep{CremerMcLean85,CremerMcLean88,krishna2009auction}.)

Unfortunately, the single crossing condition does not generally suffice to obtain optimal social welfare in
settings beyond that of a single item auction with single-dimensional signals. It is insufficient in fairly simple settings, such as two-item, two-bidder auctions with unit-demand valuations (see Section~\ref{sec:ud-lb}), or single-parameter settings with downward-closed feasibility constraints (see Section~\ref{sec:dc-single-param}).

Moreover, there are many relevant single-item settings where the single-crossing condition does not hold. For example, suppose that the signals indicate demand for a product being auctioned, agents represent firms, and one firm has a stronger signal about demand, but is in a weaker position to take advantage of that demand. A setting like this could yield valuations that do not satisfy
the single crossing condition. For a concrete example, consider the following scenario given by~\citep{maskin1992} and~\citep{dasgupta2000efficient}.

\begin{example}
Suppose that oil can be sold in the market at a price of 4 dollars per unit and two firms are competing for the right to drill for oil. Firm 1 has a fixed cost of 1 to produce oil and a marginal cost of 2 for each additional unit produced, whereas firm 2 has a fixed cost of 2 and a marginal cost of 1 for
each additional unit produced. In addition, suppose that firm 1 does a private test and discovers that the expected size of the oil reserve is $s_1$ units. Then
$v_1(s_1, s_2) = (4-2) s_1 -1 = 2s_1-1$, whereas $v_2(s_1, s_2) = (4-1) s_1 - 2 = 3s_1 -2$. These valuations don't satisfy the single-crossing condition since firm 1 needs to win when $s_1$ is low and lose when $s_1$ is high.
\end{example}

\subsection{Research Problems}

This paper addresses the following two issues related to social welfare maximization in the interdependent values model:

\begin{enumerate}
\item To what extent can the optimal social welfare be approximated in interdependent settings that do not satisfy the single-crossing condition?
\item How far beyond the single item, single-dimensional setting can we go?

Given the impossibility result of \cite{jehiel2001efficient}, we ask if it
is possible to {\em approximately} maximize social welfare in {\em combinatorial auctions with interdependent values}?

\end{enumerate}
The first question was recently considered by~\cite{eden18} who gave two examples pointing out the difficulty of approximating social welfare without single crossing.
Example \ref{ex:no-det} shows that even with two bidders and one signal, there are valuation functions for which no deterministic auction can achieve {\em any} bounded approximation ratio to optimal social welfare.

 \begin{example}[No bound for deterministic auctions \cite{eden18}] \label{ex:no-det}
	A single item is for sale. There are two players, $A$ and $B$, only $A$ has a signal $s_A\in \{0,1\}$. The valuations are
	\begin{eqnarray*}
		v_A(0) = 1 &\quad v_B(0)=0\\
		v_A(1) = 2 &\quad v_B(1)=H,
	\end{eqnarray*}
where $H$ is an arbitrary large  number.
If $A$ doesn't win when $s_A=0$, then the approximation ratio is infinite. On the other hand, if $A$ does win when $s_A=0$, then by monotonicity, $A$ must also win at $s_A=1$, yielding a $2/H$ fraction of the optimal social welfare.
\end{example}

The next example can be used to show that there are valuation functions for which no randomized auction performs better (in the worst case) than allocating to a random bidder ({\sl i.e.}, a factor $n$ approximation to social welfare), even if a prior over the signals is known.

\begin{example}[$n$ lower bound for randomized auctions \cite{eden18}] \label{ex:no-rand}
	There are $n$ bidders $1,\ldots, n$ that compete over a single item. For every agent $i$, $s_i\in \{0,1\}$, and \[v_i(\bs)=\prod_{j\neq i}s_j + \epsilon \cdot s_i\quad \mbox{for $\epsilon\rightarrow 0$};\]
	that is, agent $i$'s value is high if and only if all other agents' signals are high simultaneously.
	When all signals are 1, then in any feasible allocation, there must be an agent $i$ which is allocated with probability of at most $1/n$. By monotonicity, this means that the probability this agent is allocated when the signal profile is $\bs'=(\mathbf{1}_{-i},0_i)$ is at most $1/n$ as well. Therefore, the achieved welfare at signal profile $\mathbf{s'}$ is at most $1/n+(n-1)\cdot\epsilon$, while the optimal welfare is $1$, giving a factor $n$ gap \footnote{ \cite{eden18} show that there exists a prior for which the $n$ gap still holds, \textit{even} if the mechanism knows the prior.}.
\end{example}

Therefore, {\em some} assumption is needed if we are to get good approximations to social welfare. The approach taken by~\cite{eden18} was to define a relaxed notion of single-crossing that they called $c$-single crossing and then provide mechanisms that approximately maximize social welfare, where the approximation ratio depends on $c$ and $n$, the number of agents.

In this paper, we go in a different direction, starting with the observation that
in Example~\ref{ex:no-rand}, the valuations treat the signals as highly-complementary--one has a value bounded away from zero only if all other agent's signals are high simultaneously. This suggests that the case where the valuations treat the signals more like ``substitutes" might be easier to handle.

We capture this by focusing on {\emph submodular over signals (SOS) valuations}. This means that for every $i$ and $j$,
when signals $\bs_{-j}$ are lower, the sensitivity of the valuation $v_i (\bs)$  to changes in $s_j$ is higher.
Formally, we assume that for all $j$, for any $s_j$, $\delta \geq 0$, and for any  $\bs_{-j}$ and $\bs'_{-j}$  such that
component-wise $\bs_{-j} \le \bs'_{-j}$, it holds that
$$v_i (s_j + \delta, \bs_{-j}) - v_i (s_j, \bs_{-j}) \ge v_i (s_j + \delta, \bs'_{-j}) - v_i (s_j, \bs'_{-j}).$$

Many valuations considered in the literature on interdependent valuations are SOS (though this term is not used) \cite{milgrom1982theory,dasgupta2000efficient,klemperer1998auctions}.
The simplest (yet still rich) class of SOS valuations are
{\em fully separable} valuation functions~\footnote{ This type of valuation function is ubiquitous in the economics literature on inderdependent settings; often with the function
simply assumed to be a linear function of the signals (see, e.g., \cite{jehiel2001efficient,klemperer1998auctions}).},
where there are {\em arbitrary} (weakly increasing) functions $g_{ij} (s_j)$
for each pair of bidders $i$ and $j$ such that
$$v_i (\bs) = \sum_{j=1}^n g_{ij}(s_j).$$

A more general class of SOS valuation functions are functions
of the form $v_i(\bs) = f(\sum_{j=1}^n g_{ij}(s_j))$, where $f$ is a weakly increasing concave function.

We can now state the main question we study in this paper: \emph{to what extent can social welfare be approximated in interdependent settings with SOS valuations?}  Unfortunately, Example \ref{ex:no-det} itself describes SOS valuations, so no deterministic auction can achieve any bounded approximation ratio, even for this subclass of valuations. Thus, we must turn to randomized auctions.

\subsection{Our Results and Techniques}

All of our positive results concern the design of {\em randomized, prior-free, universally ex-post incentive-compatible (IC), individually rational (IR) mechanisms}. Prior-free means that the rules of the mechanism makes no use of the prior distribution over the signals, thus need not have any knowledge of the prior.

Our first result provides approximation guarantees for single-parameter downward-closed settings.
An important special case of this result is single-item auctions, which was the focus of~\cite{eden18}.

\vspace{0.1in}
\noindent {\bf Theorem \ref{thm:single-param-dc}} (See Section \ref{sec:single-param}): For every single-parameter downward-closed setting, if the valuation functions are SOS, then the \texttt{Random Sampling Vickrey} auction is a universally ex-post IC-IR mechanism that gives a 4-approximation to the optimal social welfare.
\vspace{0.1in}

Interestingly, no deterministic mechanism can give better than an $(n-1)$-approximation for arbitrary downward-closed settings, even if the valuations are single crossing, and this is tight. Recall that for a single item auction, or even multiple identical items, with single crossing valuations, the deterministic generalized Vickrey auction obtains the optimal welfare \cite{maskin1992,ausubel1999generalized}.

\vspace{0.1in}

We then turn to multi-dimensional settings. In the most general combinatorial auction model that we consider, each agent $i$ has a signal $s_{iT}$ for each subset $T$ of items, and a valuation function $v_{iT}: = v_{iT} (s_{1T}, s_{2T}, \ldots, s_{nT})$. For this setting, it is not at all clear under what conditions it might be possible to maximize social welfare in ex-post equilibrium.\footnote{ See the related work and also Lemmas \ref{lem:sclowertwodet} and \ref{lem:sclowertworan}, which show that under one natural generalization of single-crossing to the setting of two items and two agents that are unit demand, single crossing is not sufficient for full efficiency.}

However, rather surprisingly (see the related work section below), for the case of \textit{separable SOS} valuations\footnote{ A valuation is separable-SOS if the valuation for an agent can be split into two parts, an SOS function of all other signals and an arbitrary function of the agents' own signal.  Such valuations generalize the fully separable case discussed above. See definition~\ref{def:sepvaluation}}, we are able to extend the 4-approximation guarantee to combinatorial auctions.

\vspace{0.1in}
\noindent {\bf Theorem \ref{thm:rs-vcg}} (See Section \ref{sec:Combinatorial}): For every combinatorial auction, if the valuation functions are separable-SOS, then the {\texttt{Random Sampling VCG} auction} is a universally ex-post IC-IR mechanism that gives a 4-approximation to the optimal social welfare.
\vspace{0.1in}


Finally, we consider combinatorial auctions where each agent $i$  has a single-dimensional signal $s_i$, but where the valuation function $v_{iT}$ for each subset of items $T$ is an {\em arbitrary} SOS valuation function $v_{iT}(s_1, \ldots, s_n)$. For this case, we show the following:

\vspace{0.1in}
\noindent {\bf Theorems \ref{thm:k-sig-sm}  and \ref{thm:k-sig-ssm}} (See Sections \ref{sec:k-sig-sm} and \ref{sec:k-sig-ssm}): Consider combinatorial auctions with single-dimensional signals,
where each signal takes one of $k$ possible values.  If the valuation functions are SOS, then there exists a universally ex-post IC-IR mechanism that gives a $(k+3)$-approximation to the optimal social welfare. If the valuations are strong-SOS~\footnote{See definition \ref{def:ssm}.}, the approximation ratio improves to $O(\log k)$.
\vspace{0.1in}


All of the above results, as well as our lower bounds, are summarized in Table~\ref{tab:results}. In addition,
all of the results in this paper generalize easily, with a corresponding degradation in the approximation ratio, to the weaker requirement
of $d$-SOS valuations~\footnote{ A valuation function is $d$-SOS if for all $j$, for all $\delta > 0$, and for any  $\bs_{-j}$ and $\bs'_{-j}$  such that
component-wise $\bs_{-j} \le \bs'_{-j}$, it holds that
$d \cdot \left(v_i (s_j + \delta, \bs_{-j}) - v_i (s_j, \bs_{-j})\right) \ge v_i (s_j + \delta, \bs'_{-j}) - v_i (s_j, \bs'_{-j}).$}.

\subsubsection{Intuition for results}

\vspace{0.1in}
The fundamental tension in settings with interdependent valuations that is not present in the private values setting is the following. Consider, for example, a single item auction setting where agent 1's truthful report of her signal  increases agent 2's {\em value}. Since, this increases the chance that agent 2 wins and may decrease agent 1's chance of winning, it might motivate agent 1 to strategize and misreport.

Our approach is to simply {\em prevent} this interaction. Without looking at the signals, our mechanism randomly divides the agents into two sets\footnote{ as in  \citep{GHW01}.}: potential winners and certain losers. Losers never receive any allocation.  When estimating the value of a potentially winning agent $i$, we use only the signals of losers and $i$'s own signal(s). Thus, potential winners can not impact the estimated values and hence allocations of other potential winners.  This resolves the truthfulness issue.  The remaining question is: can we get sufficiently accurate estimates of the agents' values when we ignore so many signals?

 The key lemma ({\bf Lemma~\ref{lem:rs-value}} Section \ref{sec:key}) shows that we can do so, when the valuations are SOS.  Specifically, for any agent $i$, if all agents other than $i$ are split into two random sets $A$ (losers) and $B$ (potential winners), and the signals of agents in the random subset $B$ are ``zeroed out'', then the expected value agent $i$ has for the item is at least half of her true valuation. That is,
$$E_A [v_i (s_i, \bs_A, \bzero_B)] \ge \frac{1}{2}v_i (\bs).$$

Dealing with combinatorial settings is more involved as the truthfullness characterization is less obvious, but the key ideas of random partitioning and using the signals of certain losers remain at the core of our results.

\subsubsection{Additional remarks}

While this paper deals entirely with welfare maximization, our results have significance for the objective of maximizing the seller's revenue. \cite{eden18} give a reduction from revenue maximization to welfare maximization in single-item auctions with SOS valuations. Thus, the constant factor approximation mechanism presented in this paper implies a constant factor approximation to the optimal revenue in single-item auctions with SOS valuations. We note that this is the first revenue approximation result that does not assume any single-crossing type assumption (\citep{CFK,eden18,RTCoptimalrev,li2016approximation} require single crossing or approximate single crossing).

Finally, one can easily verify that, based on Yao's min-max theorem, the existence of a {\em randomized prior-free mechanism} that gives some approximation guarantee (in expectation over the coin flips of the mechanism) implies the existence of a {\em deterministic prior-dependent} mechanisms that gives the same approximation guarantee (in expectation over the signal profiles).

\begin{table}[]
\begin{tabular}{|l|l|}
\hline
\multicolumn{1}{|c|}{Setting}                                                                                                                   & \multicolumn{1}{c|}{Approximation Guarantees}                                                                 \\ \hline
\begin{tabular}[c]{@{}l@{}}Single Parameter SOS valuations \\ Downward Closed Feasibility\\ Single-Dimensional Signals\end{tabular}             & \begin{tabular}[c]{@{}l@{}}$\geq 1/4$\\ $\forall \mathrm{mech.\ } \leq 1/2$ \qquad \qquad (Section \ref{sec:single-param})\end{tabular}           \\ \hline
\begin{tabular}[c]{@{}l@{}}Arbitrary Combinatorial SOS valuations\\ Single-Dimensional Signals, $k$-sized Signal Space\end{tabular}         & \begin{tabular}[c]{@{}l@{}}$\geq 1/(k+3)$\\ $\forall \mathrm{mech.\ } \leq 1/2$ \qquad \qquad (Section \ref{sec:k-sig-sm})\end{tabular}       \\ \hline
\begin{tabular}[c]{@{}l@{}}Arbitrary Combinatorial, Strong-SOS Valuations\\ Single-Dimensional Signals, $k$-sized Signal Space\end{tabular} & \begin{tabular}[c]{@{}l@{}}$\geq 1/(\log(k)+2)$\\ $\forall \mathrm{mech.\ } \leq 1/2$ \qquad \qquad (Section \ref{sec:k-sig-ssm})\end{tabular} \\ \hline
\begin{tabular}[c]{@{}l@{}}Combinatorial, Separable-SOS Valuations\\ Multi-Dimensional Signals\end{tabular}                                 & \begin{tabular}[c]{@{}l@{}}$\geq 1/4$\\ $\forall \mathrm{mech.\ }\leq 1/2$ \qquad \qquad (Section \ref{sec:Combinatorial})\end{tabular}           \\ \hline
\end{tabular}
\caption{The table shows the approximation factors achievable for social welfare maximization with SOS and strong-SOS valuations.  Similar results hold for $d$-approximate SOS/Strong-SOS valuations, while losing a factor that depends on $d$. All positive results are obtained with universally ex-post IC-IR randomized mechanisms.}
\label{tab:results}
\end{table}

\subsection{More on Related Work}
\label{sec:related}

As discussed above, in single-parameter settings, there is an extensive literature on mechanism design with interdependent valuations that considers social welfare maximization, revenue maximization and other objectives. However, the vast majority of this literature assumes some kind of single-crossing condition and, in the context of social welfare, focuses on exact optimization.

There are two papers that we are aware of that study the question of how well optimal social welfare can be approximated in ex-post equilibrium without single-crossing. The first is the aforementioned paper \citep{eden18} on single item auctions with interdependent valuations. They defined a parameterized version of single-crossing, termed $c$-single crossing, where $c>1$ is a parameter that indicates how close is the valuation profile to satisfy single-crossing.
For $c$-single crossing valuations, they provide a number of results including a lower bound of $c$ on the approximation ratio achievable by any mechanism, a matching upper bound for binary signal spaces, and mechanisms that achieve approximation ratios of $(n-1)c$ and $2c^{3/2}\sqrt{n}$  (the first is deterministic and the second is randomized).

\cite{ItoP06} also consider approximating social welfare in the interdependent setting. Specifically, they propose a greedy contingent-bid auction (a la \citep{dasgupta2000efficient}) and show that it achieves a $\sqrt{m}$ approximation to the optimal social welfare for $m$ goods, in the special case of combinatorial auctions with single-minded bidders.

For multidimensional signals and settings, the landscape is sparser (and bleaker) and, to our knowledge, focuses on exact social welfare maximization. \cite{maskin1992} has observed that, in general, no efficient incentive-compatible single item auction exists if a buyer's valuation depends on a multi-dimensional signal.

\citet{jehiel2001efficient} 
consider a very general model in which there is a set $K$ of possible alternatives, and a multidimensional signal space, where each agent $j$ has a signal $s_{ki}^j$ for each outcome $k$ and other agent $j$.  In their model the valuation function of an agent $i$ for outcome $k$
is linear in the signals, that is, $v_i (k) := \sum_j a_{ki}^j s_{ki}^j$. Thus, their valuation functions are, in one sense, a special case of our separable valuation functions.
On the other hand, they are more general in that all quantities depend on the outcome $k$. Thus, there are allocation externalities. Their main result is that, generically, there is no Bayes-Nash incentive compatible mechanism that maximizes social welfare in this setting. However, they do give an ex-post IC mechanism that maximizes social welfare with both information and allocation externalities if
the signals are one-dimensional, the valuation functions are linear in the signals,
and a single-crossing type condition holds.

\cite{jehiel2006limits} go on to show that the only deterministic social choice functions that are ex-post implementable in generic mechanism design frameworks with multidimensional signals, interdependent valuations and transferable utilities, are constant functions.

Finally, \cite{Bikhchandani2006}  considers a single item setting with multidimensional signals but no allocation externalities and shows that there is a generalization of single-crossing that allows some social choice rules to be implemented ex-post.

For further analysis and discussion of implementation with interdependent valuations, see e.g., \cite{Bergemann05}
and \cite{McLean2015}.

For further literature in computer science on interdependent and correlated values, see
~\citep{ronen2001approximating,ConstantinIP07,ConstantinP07,KleinMPPSW08,PapadimitriouPierrakos10,DobzinskiFK11,BabaioffKL12,
abraham2011peaches,RobuPIJ13,kempe2013information,che2015efficient,li2016approximation,CFK}.

\section{Model and Definitions}

\subsection{Single Parameter Settings}

In Section \ref{sec:single-param}, we will consider single-parameter settings with interdependent valuations and downward-closed feasibility constraints.
In these settings, a mechanism decides which subset of agents  $1,\ldots,n$ are to receive ``service'' (e.g., an item). The feasibility constraint
is defined by a collection  $\mathcal{I}\subseteq 2^{[n]}$ of subsets of agents that may feasibly be served simultaneously. We restrict attention to {\em downward-closed settings}, which means that any subset of a feasible set is also feasible.  A simple example is a $k$-item auction, where $\mathcal{I}$ is the collection of all subsets of agents of size at most $k$.

For  these settings, we use the interdependent value model of~\cite{milgrom1982theory}:

\begin{definition}[Single Dimensional Signals, Single Parameter Valuations] \label{def:singlesingle}
Each agent  $j$ has a private signal $s_j\in \Realsp$.
The value agent $j$ gives to ``receiving service'' $v_j(\bs)\in \Realsp$, is a function of all agents' signals $\bs=(s_{1}, s_{2}, \ldots, s_{n})$. The function $v_j(\bs)$ is assumed to be weakly increasing in each coordinate and strictly increasing in $s_i$.
\end{definition}

\subsubsection{Deterministic Mechanisms}

\begin{definition}[Deterministic Single Parameter Mechanisms]
\label{def:SPmech}
A deterministic mechanism $M= (x,p)$ in the downward closed setting
is a mapping from reported signals $\bs = (s_1, \ldots, s_n)$
to allocations $x(\bs) = \{x_i (\bs)\} _{1\le i \le n}$ and payments
$p(\bs) = \{p_i(\bs)\}_{1\le i\le n}$, where $x_i (\bs) \in \{0,1\}$ indicates whether or not
agent $i$ receives service and $p_i(\bs)$ is the payment of agent $i$.
It is required that the set of
agents that receive service is feasible, i.e., $\{i ~|~ x_i(\bs)=1\} \in \mathcal{I}$.
(The mechanism designer knows the form of the valuation functions
but learns the private signals only when they are reported.)
 \end{definition}

\begin{definition}[Agent utility]
\label{SP:detutility}
Given a deteministic mechanism $(x,p)$, the {\em utility} of agent $i$ when her true signal is $s_i$, she reports $s_i'$ and the other agents report $\bs_{-i}$ is
$$u_i(s_i', \bs_{-i} | s_i) = x_i (s_i', \bs_{-i}) v_i ( s_i,  \bs_{-i}) - p_i (s_i', \bs_{-i}).$$
Agent $i$ will report $s_i'$ so as to maximize $u_i(s_i', \bs_{-i} | s_i)$.
We use $u_i (\bs)$ to denote the utility when she reports truthfully, i.e., $u_i(s_i, \bs_{-i} | s_i)$.
\end{definition}

\begin{definition}[Deterministic ex-post incentive compatibility (IC)]
\label{SP:detIC}
A deterministic mechanism $M = (x, p)$ in the interdependent setting is \emph{ex-post incentive compatible} (IC) if, irrespective of the true signals, and given that all other agents report their true signals, there is no advantage to an agent to report any signal other than her true signal. In other words, assuming that $\bs_{-i}$ are the true signals of other bidders,
$u_i(s_i', \bs_{-i} | s_i)$ is maximized by reporting $s_i$ truthfully.
\end{definition}

\begin{definition}[Deterministic ex-post  individual rationality (IR)]
\label{SP:detIR}
A deterministic mechanism in the interdependent setting is \emph{ex-post individually rational} (IR) if, irrespective of the true signals, and given that all other agents report their true signals, no agent gets negative utility by participating in the mechanism.
\end{definition}

If a deterministic mechanism is both ex-post IR and ex-post IR we say that it is ex-post IC-IR.

\begin{definition}
A deterministic allocation rule $x$ is monotone if for every agent $i$, every signal profile of all other agents $s_{-i}$, and every $s_i \leq s'_i$, it holds that
$x_i(s_i,s_{-i}) =1 \Rightarrow x_i(s'_i,s_{-i})=1$.
\end{definition}

\begin{proposition} \citep{RTCoptimalrev}\label{prop:monotone}
For every deterministic allocation rule $x$ for single parameter valuations, there exist payments $p$ such that the mechanism $(x,p)$ is ex-post IC-IR if and only if $x_i$ is monotone for every agent $i$.
\end{proposition}

\subsubsection{Randomized Mechanisms}

\begin{definition}
A randomized mechanism is a probability distribution over deterministic mechanisms.
\end{definition}

\begin{definition}[Universal ex-post IC-IR]
A randomized mechanism is said to be universally ex-post IC-IR if all deterministic mechanisms in the support are ex-post IC-IR.
\end{definition}

\subsection{Combinatorial Valuations with Interdependent Signals}

Sections \ref{sec:Combinatorial} and \ref{sec:comb-auction-single-dimension}  focus on combinatorial auctions, where there are $n$ agents and $m$ items. In these settings, a mechanism is used to
decide how the items are partitioned among the agents.  We consider two models for the interdependent valuations:~\footnote{ For other types of signals and interdependent valuation models, see, e.g.,~\cite{jehiel2001efficient}.}

\begin{definition}[Single Dimensional Signals, Combinatorial Valuations] \label{def:singlecomb}
 Each agent $i$ has a signal $s_i\in \Realsp$.
The value agent $i$ gives to subset of items $T\subseteq [m]$, which we denote by $v_{jT}(\bs)$, is a function of $\bs=(s_{1}, s_{2}, \ldots, s_{n})$.
\end{definition}

\begin{definition}[Multidimensional Combinatorial Signals, Combinatorial Valuations] \label{def:multicomb}
Here, each agent has a signal for each subset of items; for any agent $i$, we use $s_{iT}$ to denote agent $i$'s signal for  subset of items $T\subseteq [m]$. The value agent $i$ gives to set $T$ is denoted by $v_{iT}(\bs_T)$ where $\bs_T=(s_{1T}, s_{2T}, \ldots, s_{nT}) \in {\Realsp}^n$.
We use $\bs$ to denote the set of all signals $\{\bs_T\} _{T \subseteq 2^m}$.
\end{definition}
In both cases, each $v_{iT}(\cdot)$ is assumed to be a weakly increasing function of each signal and strictly increasing in $s_{i}$ (or $s_{iT}$ respectively), and known to the mechanism designer.

We give subsequent definitions only for multidimensional combinatorial signals, as single dimensional signals can be viewed as a special case of multi-dimensional signals where $s_{iT} = s_i$ for all $T$.

\subsubsection{Deterministic Mechanisms}

\begin{definition}[Deterministic mechanisms for combinatorial settings]
\label{def:SPmech}
A deteministic mechanism $M = (x,p)$ is a mapping from reported signals $\bs$ to allocations $x=\{x_{iT}\}$  (where each $x_{iT} \in \{0,1\}$) and payments $p=\{p_{iT}\}$ for all $1\leq i \leq n$ and $T\subset\{1,\ldots,m\}$ such that: \begin{itemize} \item Agent $j$ is allocated  the set $T$ iff $x_{jT}(\bs)=1$;
 \item For each agent $j$, there is at most one $T$ for which $x_{jT}(\bs)=1$;
 \item The sets allocated to different agents do not intersect.
 \item The payment for agent $j$ when her allocation is set $T$ is $p_{jT}(\bs)$.
 \end{itemize}
 \end{definition}

\begin{definition}[Agent Utility]
\label{SP:utility}
The {\em utility} of agent $i$ when her signals are $\bs_i = \{s_{iT}\}_{T \subset 2^m}$, she reports $\bs_i'$ and the other agents report $\bs_{-i}$ is
$$u_i(\bs_i', \bs_{-i} | \bs_i) = \sum_{T \subseteq 2^m} x_{iT} (\bs_i', \bs_{-i}) [v_{iT} ( \bs_{iT},  \bs_{-iT}) - p_{iT}(\bs'_{i}, \bs_{-i})].$$
Given a mechanism $M= (x,p)$, agent $i$ will report $\bs_i'$ so as to maximize $u_i(\bs_i', \bs_{-i} | \bs_i) $.
We use $u_i (\bs)$ to denote the utility when she reports truthfully, i.e., $u_i(\bs_i, \bs_{-i} | \bs_i)$.
\end{definition}

The definitions of \emph{ex-post incentive compatibility} (IC) and
\emph{ex-post individually rationality} (IR) for deterministic mechanisms for combinatorial settings are the same as the appropriate definitions for single parameter mechanisms
(Definitions~\ref{SP:detIC} and ~\ref{SP:detIR} with the obvious modifications).

\subsubsection{Randomized Mechanisms}

As with single parameter mechanisms, a randomized mechanism for a combinatorial setting is a probability distribution over deterministic mechanisms for the combinatorial setting, and a randomized mechanism is said to be \emph{universally ex post IC-IR} if all deterministic mechanisms in the support are themselves ex-post IC-IR.

\subsection{Submodularity over signals (SOS)} \label{sec:submodularity}
 \label{sec:SOS}

As discussed in the introduction, our results will rely on an assumption about the valuation functions that we call {\em submodularity over signals} or SOS.
The SOS (resp. strong-SOS) notion we use is the same as the weak diminishing returns (resp. strong diminishing returns) submodularity notion in~\citep{Bian17, Niazadeh18}\footnote{ Weak diminishing returns submodularity was introduced in \citep{Soma_2015}, where it's termed ``diminishing returns submodularity''. }. SOS was also used in \citep{eden18}, generalizing a similar notion in \citep{CFK}.

\begin{definition}[$d$-approximate submodular-over-signals valuations ($d$-SOS valuations)] \label{def:sm}
A valuation function $v (\bs)$ is a
 \emph{$d$-SOS valuation} if for all $j$, $s_{j}$, $\delta\geq 0$, $$\mathbf{\bs}_{-j}=(s_{1}, \ldots, s_{j-1},s_{j+1}, \ldots, s_{n})\quad\text{ and }\quad\mathbf{\bs}'_{-j}=(s'_{1}, \ldots, s'_{j-1},s'_{j+1}, \ldots, s'_{n})$$ such that $\mathbf{\bs}'_{-j}$ is smaller than or equal to $\mathbf{\bs}_{-j}$ coordinate-wise, it holds that
\begin{equation} d \cdot \left(v(\mathbf{\bs}'_{-j}, s_{j}+ \delta) - v(\mathbf{\bs}'_{-j},s_{j})\right) \geq v(\mathbf{\bs}_{-j}, s_{j}+ \delta) - v(\mathbf{\bs}_{-j},s_{j})\label{eq:sm}
\end{equation}
If $v$ satisfies this condition with $d=1$, we say that $v$ is an SOS valuation function.
\end{definition}

\begin{definition}[$d$-approximate strong submodular-over-signals valuations ($d$-strong-SOS valuations)] \label{def:ssm}
The valuation function $v(\bs)$ is a
 \emph{$d$ strong-SOS valuation} if for any  $j$, $\delta\geq 0$, $$\mathbf{\bs}=(s_{1},  \ldots, s_{n})\quad\text{ and }\quad\mathbf{\bs}=(s'_{1}, \ldots, s'_{n})$$ such that $\mathbf{\bs}'$ is smaller than or equal to $\mathbf{\bs}$ coordinate-wise, it holds that
\begin{equation} d \cdot \left(v(\mathbf{\bs}'_{-j}, s'_{j}+ \delta) - v(\mathbf{\bs}'_{-j},s'_{j})\right) \geq v(\mathbf{\bs}_{-j}, s_{j}+ \delta) - v(\mathbf{\bs}_{-j},s_{j})\label{eq:sm}
\end{equation}
If $v$ satisfies this condition with $d=1$, we say that $i$'s valuation functions are ``strong-SOS''.
\end{definition}

\begin{definition}[SOS-valuations settings]
We say that a mechanism design setting with interdependent valuations is an \emph{SOS-valuations setting} or, equivalently, that the agents have SOS-valuations, in each of the following cases:
\begin{itemize}
\item Single parameter valuations (as in definition \ref{def:singlesingle}): for every $i$, the valuation function $v_i(\bs)$ is SOS.
\item Combinatorial valuations with single-parameter signals (as in definition \ref{def:singlecomb}): for every $i$ and $T$, the valuation function $v_{iT} (\bs)$ is SOS;
\item Combinatorial valuations with multi-parameter signals (as in definition \ref{def:multicomb}): for every $i$ and $T$,
$v_{iT}(\bs_T)$ is SOS, where $ \bs_T = (s_{1T}, \ldots, s_{nT})$.
\end{itemize}
Similar definitions can be given for $d$-SOS valuation settings
and $d$-strong-SOS valuation settings.
\end{definition}

Finally, in section \ref{sec:Combinatorial}, we will specialize to the case of \emph{separable} SOS valuations.

\begin{definition}[Separable SOS valuations] \label{def:sepvaluation}
We say that a set of valuations as in Definition \ref{def:multicomb} are \emph{separable SOS valuations} if for every agent $i$ and subset $T$ of items, $v_{iT}(\bs_T)$ can be written
as $$v_{iT} (\bs_T) =  g_{-iT}( \bs_{-iT}) + h_{iT} (s_{iT}) ,$$
where $g_{-iT}(\cdot)$ and $h_{iT}(\cdot)$ are both weakly increasing
and $g_{-iT}(\bs_{-iT})$ is itself an SOS valuation function.
\end{definition}

\begin{observation}
  A separable SOS valuation function is itself an SOS valuation function.
\end{observation}

We can similarly define separable $d$-SOS valuations.

\subsection{A useful fact about SOS valuations}

\begin{lemma}\label{lem:sm-sets}
	Let $v: {\Realsp}^{n} \rightarrow \Realsp$ be a $d$-SOS function. Let $A\subseteq [n]$ and $B=[n]\setminus A$. For any $\mathbf{s}_A, \mathbf{y}_A\in {\Realsp}^{|A|}$, and $\mathbf{s}_B, \mathbf{s'}_B\in {\Realsp}^{|B|}$ such that $\mathbf{s}_B$ is smaller than $\mathbf{s'}_B$ coordinate wise,
	\begin{eqnarray*} d \cdot \left(v(\mathbf{s}_{A}+\mathbf{y}_{A},\mathbf{s}_{B}) - v(\mathbf{s}_{A},\mathbf{s}_{B})\right) \geq v(\mathbf{s}_{A}+\mathbf{y}_{A},\mathbf{s'}_{B}) - v(\mathbf{s}_{A},\mathbf{s'}_{B}).
	\end{eqnarray*}
\end{lemma}
\begin{proof}
	Let $i_1, i_2, \ldots, i_{|A|}$ be the elements of $A$. For $1\leq j\leq |A|$, let $\mathbf{s}^j$ and $\mathbf{s'}^j$ denote the vectors
	\begin{eqnarray*}
		\mathbf{s}^j & = & \left(({s}_{i_1}+{y}_{i_1}),  \ldots, ({s}_{i_j}+{y}_{i_j}), s_{i_{j+1}},\ldots, s_{i_{|A|}},\mathbf{s}_{B}\right),\\
		\mathbf{s'}^j & = & \left(({s}_{i_1}+{y}_{i_1}),  \ldots, ({s}_{i_j}+{y}_{i_j}), s_{i_{j+1}},\ldots, s_{i_{|A|}},\mathbf{s'}_{B}\right).
	\end{eqnarray*}
Note that $\mathbf{s}^{|A|}=(\mathbf{s}_{A}+\mathbf{y}_{A},\mathbf{s}_{B})$, and $\mathbf{s'}^{|A|}=(\mathbf{s}_{A}+\mathbf{y}_{A},\mathbf{s'}_{B})$.
	
	It follows from the $d$-SOS definition that for every $1\leq j\leq |A|$,
	\begin{eqnarray}
		d\cdot\left( v(\mathbf{s}^j)- v(\mathbf{s}^{j-1}) \right) \geq v(\mathbf{s'}^j)- v(\mathbf{s'}^{j-1}), \label{eq:step}
	\end{eqnarray}
	where $\mathbf{s}^0= (\mathbf{s}_{A},\mathbf{s}_{B})$ and $\mathbf{s'}^0= (\mathbf{s}_{A},\mathbf{s'}_{B})$.
	
	Summing Equation~\eqref{eq:step} for $j=1,2,\ldots, |A|$ proves the claim.
\end{proof}

\section{The Key Lemma}\label{sec:key}

The following is a key lemma which is used for both single parameter and combinatorial settings.

\begin{lemma}\label{lem:rs-value}
Let $v_i: {\Realsp}^{n} \rightarrow \Realsp$ be a $d$-SOS function.
Let $A$ be a uniformly random subset of $[n] \setminus \{i\}$, and let $B := \left([n]\setminus \{i\}\right)\setminus A $.
It now holds that
$$\E_A\left[v_i(\mathbf{s}_A, \mathbf{0}_B, s_i)\right] \geq \frac{1}{d+1}v_i(\mathbf{s}),$$ where the expectation is over the random choice of $A$.
\end{lemma}

\begin{proof}
  We consider two equiprobable events,
  \begin{itemize}
    \item $A= S \subset [n]\setminus \{i\}$ is chosen as the random subset.
    \item $A = T = \left([n] \setminus \{i\}\right) \setminus S$ is chosen as the random subset.
  \end{itemize}
Normalize the valuations so that $v_i(\mathbf{s})= 1$ and define $\alpha, \beta \in [0,1]$ such that  \begin{eqnarray*}
v_i(\mathbf{s}_S,\mathbf{0}_T, s_i)\ =\ \alpha, \quad\quad
v_i(\mathbf{0}_S,\mathbf{s}_T,s_i)\ =\ \beta.\end{eqnarray*}
It follows that
\begin{eqnarray*}
\beta = v_i(\mathbf{0}_S,\mathbf{s}_T, s_i) &\geq & v_i(\mathbf{0}_S,\mathbf{s}_T, s_i) - v_i(\mathbf{0}_S,\mathbf{0}_T, s_i)  \\
&\geq & (v_i(\mathbf{s}_S,\mathbf{s}_T, s_i)-v_i(\mathbf{s}_S,\mathbf{0}_T, s_i))/d\\
& = & (1-\alpha)/d,
\end{eqnarray*}
where the first inequality follows from non-negativity of $v_i(\mathbf{0}_S,\mathbf{0}_T, s_i)$, and the second inequality follows from $v_i$ being $d$-SOS and Lemma~\ref{lem:sm-sets}.

Similarly, we have that
\begin{eqnarray*}
  \alpha &\geq& (1-\beta)/d \quad \Rightarrow \quad \beta \geq 1-\alpha d;
\end{eqnarray*}


It follows that  $$\alpha + \beta \geq \max \left( \alpha+\frac{1-\alpha}{d}, \alpha + 1 - \alpha d\right).$$
Solving for equality of the two terms, we get that $\alpha = 1/(d+1)$ which implies that
$$\alpha + \beta \geq \frac{2}{d+1}.$$

Partition the event space into pairs $(S,T)$ that partition $[n]\setminus \{i\}$.
For every such $(S,T)$ pair, it follows that $v_i(\mathbf{s}_S,\mathbf{0}_T, s_i)+v_i(\mathbf{0}_S,\mathbf{s}_T,s_i) = \alpha + \beta \geq \frac{2}{d+1}$.

We conclude with the following, where the third line follows from the fact that there are $2^{n-1}/2$ such $(S,T)$ pairs that partition $[n] \setminus \{i\}$:
	\begin{eqnarray*}
		\E_A \left[ v_i(\mathbf{s}^A,\mathbf{0}_B,s_i) \right] & = &\sum_{A \subseteq [n] \setminus \{i\}} \Pr[A]\cdot v_i(\mathbf{s}^A,\mathbf{0}_B,s_i)\\
		&= &\frac{1}{2^{n-1}}\cdot \sum_{A \subseteq [n] \setminus \{i\}} v_i(\mathbf{s}^A,\mathbf{0}_B, s_i)\\
		&\geq&\frac{1}{2^{n-1}}\cdot \frac{2^{n-1}}{2}\cdot \frac{2}{d+1}\ =\ \frac{1}{d+1},
	\end{eqnarray*}
	as desired.
\end{proof}

\section{Single-Parameter Valuations}
\label{sec:single-param}


In this section we describe the \texttt{Random Sampling Vickrey (RS-V) mechanism} that achieves a 4-approximation for single-parameter downward-closed environments with SOS valuations and a $2(d+1)$-approximation for $d$-SOS valuations.  
We then give a lower bound of 2 and $\sqrt{d}$ for SOS and $d$-SOS valuations respectively, even in the case of selling a single item.

Let $\mathcal{I}\subseteq 2^{[n]}$ be a downward-closed set system. We present a mechanism that serves only sets in $\mathcal{I}$ and gets a $2(d+1)$-approximation to the optimal welfare.\\

\noindent \texttt{Random Sampling Vickrey (RS-V):}

\begin{itemize}
	\item Elicit bids $\mathbf{\tilde{s}}$ from the agents.
	\item Partition the agents into two sets, $A$ and $B$,  uniformly at random.
	\item For $i \in B$, let $w_i = v_i(\mathbf{\tilde{s}}_A, \tilde{s}_i, \mathbf{0}_{B\setminus \{i\}})$.
	\item Allocate to a set of bidders in $${\argmax}_{S\in \mathcal{I}\ :\ S\subseteq B}\ \left\{\sum_{i\in S}w_i.\right\}$$		
\end{itemize}

\begin{theorem}	\label{thm:single-param-dc}
For agents
with SOS valuations, and for every downward-closed feasibility constraint $\mathcal{I}$,  \texttt{RS-V} is an ex-post IC-IR mechanism that gives $4$-approximation to the optimal welfare. 
For $d$-SOS valuations, the mechanism gives a $2(d+1)$-approximation to the optimal welfare.
\end{theorem}
\begin{proof}
	We first show the allocation is monotone in one's signal, and hence, by Proposition~\ref{prop:monotone}, the mechanism is ex-post IC-IR. Fix a random partition $(A,B)$.
	\begin{itemize}
		\item Agents in $A$ are never allocated anything and thus their allocation is weakly monotone in their signal.
		\item For an agent $i\in B$, increasing $\rep_i$ can only increase $w_i$, whereas it leaves $w_j$ unchanged for all $j\in B \setminus \{i\}$. Thus, this only increases the weight of feasible sets (subsets of $B$ in $\mathcal{I}$) that $i$ belongs to. Therefore, increasing $s_i$ can only cause $i$ to go from being unallocated to being allocated.
	\end{itemize}

	
	For approximation, consider a set $S^*\in \argmax_{S\in \mathcal{I}}\sum_{i\in S}v_i(\mathbf{s})$ that maximizes social welfare. For every $i\in S^*$, from the Key Lemma~\ref{lem:rs-value}, we have that
	\begin{eqnarray}
	\E _B [w_{i} \cdot \mathbf{1}_{i \in B}]  = \E _B [v_{i}(\bs_i, \bs_A, \bzero_{B_{-i}}) ~|~ i \in B] \cdot Pr (i \in B)
	\ge \frac{ v_{i}(\bs)}{d+1}\cdot  \frac{1}{2}.\label{eq:rs-v-dc-sm}
	\end{eqnarray}
	
	For every set $B$, the fact that $\mathcal{I}$ is downward-closed implies that $S^*\cap B\in \mathcal{I}$. Therefore, $S^*\cap B$ is eligible to be selected by  \texttt{RS-V} as the allocated set of bidders. We have that the values of the bidders we allocate to are at least
	\begin{eqnarray*}
		\E_B \left[\max_{S\in \mathcal{I}: S\subseteq B}\sum_{i\in S}w_i\right] & \geq &  \E_B \left[\sum_{i\in S^*\cap B}w_i\right]
		\ =\  \E _B \left[\sum_{i\in S^*} w_{i} \cdot \mathbf{1}_{i \in B}\right]\\
		& =& \sum_{i\in S^*}\E _B \left[ w_{i} \cdot \mathbf{1}_{i \in B}\right]\ \geq \  \sum_{i\in S^*}\frac{ v_{i}(\bs)}{2(d+1)},
	\end{eqnarray*}
	as desired.  Since the allocated bidders' true values at $\mathbf{s}$ are only higher than the proxy values $w_i$, this continues to hold.
	
\end{proof}

We note that for the case of downward-closed feasibility constraints, even if the valuations satisfy single-crossing, there can be an $n-1$ gap between the optimal welfare and the welfare that the best deterministic mechanism can get. This is stated in Theorem~\ref{thm:dc-sc-lb} in Section~\ref{sec:dc-single-param}.

The following lower bounds, Theorem \ref{thm:noexpost} show that even for a single item setting, one cannot hope to get a better approximation than $2$ and $\Omega({\sqrt{d}})$ for SOS and $d$-SOS valuations respectively.
 The lower bounds apply to arbitrary randomized mechanisms\footnote{ A randomized mechanism takes as input the set of signals $\bs$
and produces as output $x_i (\bs)$ and $p_i(\bs)$ for each agent $i$,
where $x_i(\bs)$ is the probability that agent $i$ wins and $p_i(\bs)$ is agent $i$'s expected payment. Such a mechanism is ex-post IC (but not necessarily universally so) if and only if $x_i(s_i, \bs_{-i})$ is monotonically increasing in $s_i$.}.

\begin{theorem}\label{thm:noexpost}
	No ex-post IC-IR mechanism (not necessarily universal) for selling a single item can get a better approximation than
	\begin{enumerate}
		\item [(a)] a factor of 2 for SOS valuations.
		\item [(b)] a factor of $\Omega(\sqrt{d})$ for $d$-SOS valuations.
	\end{enumerate}
\end{theorem}
\begin{proof}
	Let $x_i(\mathbf{s})$ be the probability agent $i$ is allocated at signal profile $\mathbf{s}$. Notice that for every $\mathbf{s}$, $\sum_i x_i(\mathbf{s})\leq 1$, otherwise the allocation rule is not feasible.
	\begin{enumerate}
		\item [(a)] Consider the case where there are two agents, 1 and 2, $s_1\in\{0,1\}$ and agent~2 has no signal. The valuations are $v_1(0)=1$, $v_1(1)=1+\epsilon$, $v_2(0)=0$ and $v_2(1)=H$ for $H\gg 1\gg \epsilon$. It is easy to see the valuations are SOS.
		
		In order to get better than a 2-approximation at $s_1 = 0$, we must have $x_1(0)>1/2$.  By monotonicity, this forces $x_1(1)>1/2$ as well, and hence $x_2(1)<1/2$ by feasibility. This implies that the expected welfare when $s_1=1$ is $x_1(1)v_1(1)+x_2(1)v_2(1) < H/2+1$, while the optimal welfare when $s_1 = 1$ is $H$. For a large $H$, this approaches a 2-approximation.  Note that this lower bound applies even given a known prior distribution on the signals in the event that we have a prior on the signals that satisfies: $\Pr[s_1=0]\cdot 1=\Pr[s_1=1]\cdot H$.
		\item [(b)] Consider the case where there are $n=\sqrt{d}$ agents and $s_i\in \{0,1\}$ for every agent $i$. The valuation of agent $i$ is
		\begin{eqnarray*}
			v_i(\mathbf{s})=\begin{cases}
				\sum_{j\neq i} s_j +\epsilon\cdot s_i\quad &\exists j\neq i\ :\ s_j =0\\
				d +\epsilon\cdot s_i\quad &s_j=1 \ \forall j\neq i,
			\end{cases}
		\end{eqnarray*}
		where $\epsilon\rightarrow 0$.
		
		To see that the valuations are $d$-SOS, notice that whenever a signal $s_j$ changes from~0 to~1, the valuation of agent $i\neq j$ increases by~1 \textit{unless} all other signals beside $i$'s are already set to $1$, in which case the valuation increases by $d-\sqrt{d}+2 < d$.
		Consider valuation profiles $\mathbf{s^i}=(0_i, \mathbf{1}_{-i})$. Note that by monotonicity, for every truthful mechanism, it must be the case that $x_i(\mathbf{s^i})\leq x_i(\mathbf{1})$. Since any feasible allocation rule must satisfy $\sum_{i=1}^{\sqrt{d}}x_i(\mathbf{1})\leq 1$, then it must be the case there exists some agent $i$ such that $x_i(\mathbf{1})\leq \frac{1}{\sqrt{d}}$, which by monotonicity implies that $x_i(\mathbf{s^i})\leq \frac{1}{\sqrt{d}}$. However, at profile $\mathbf{s^i}$, $v_i(\mathbf{s^i})=d$ while $v_j(\mathbf{s^i})=\sqrt{d}-2< \sqrt{d}$ for all $j\neq i$, so we get that the expected welfare of the mechanism at $\mathbf{s^i}$ is at most $x_i(\mathbf{s^i})\cdot d + (1-x_i(\mathbf{s^i}))\cdot \sqrt{d} \leq 2\sqrt{d},$ while the optimal welfare is $d$. Again, the lower bound also applies to the setting with known priors on the signals using a prior that satisfies: $\Pr[\mathbf{s^i}]=\Pr[\mathbf{s^j}]=\frac{1}{\sqrt{d}}$ for all $i$ and $j$.
	\end{enumerate}
\end{proof}

\section{Combinatorial Auctions with Separable Valuations}\label{sec:Combinatorial}

In this section we present an ex-post IC-IR mechanism that gives $1/4$ of the optimal social welfare in any combinatorial auction setting with separable SOS valuations  (as in Definition \ref{def:sepvaluation}). The mechanism, that we call the \texttt{Random-sampling VCG} auction is a natural extension of the \texttt{Random-Sampling Vickrey (RS-V)} auction presented in Section~\ref{sec:single-param}. Note that unlike \texttt{RS-V}, here we need to explicitly define payments so that the obtained mechanism is ex-post IC-IR. We derive VCG-inspired payments which align the objective of the mechanism with that of the agents. Separability is used here, as without it, the payment term would have been affected by the agent's report (while with separability, only the allocation is affected by it).


%


\vspace{0.2in}
\noindent \texttt{Random-Sampling VCG (RS-VCG):}
\begin{itemize}
	\item Agents report their signals $\tilde \bs$.
	\item Partition the agents into two sets $A$ and $B$ uniformly at random.
	\item For each agent $j\in B$ and bundle $T\subseteq[m]$, let $$w_{jT}:= v_{jT} ( \tilde s_{jT},\tilde\bs_{AT}, \bzero_{B_{-j}T}) =g_{-jT} (\tilde\bs_{AT},\bzero_{B_{-j}T}) + h_{jT} ( \tilde s_{jT}) .$$
	\item Let the allocation be $$\{T_i\}_{i\in B}\in \argmax_{\{S_i\}_{i\in B}}\sum_{i\in B}w_{iS_i};$$ \textit{i.e.,}  $\{T_i\}_{i\in B}$ is the allocation that maximizes the ``welfare'' using $w_{iT}$'s.
	\item Set the payment for a winning agent $i\in B$ receiving set of goods $T_i$ to be:
	$$p_i (\tilde\bs): = g_{-iT_i} (\tilde\bs_{-iT_i}) - g_{-iT_i} (\tilde\bs_{AT_i},\bzero_{B_{-i}T_i}) - \sum_{j \in B \setminus \{i\}} w_{jT_j} + w_{-i}, $$
	where $$w_{-i} =\max_ {\text{partitions }\{T'_j\} }\sum_{j \in B \setminus \{i\}} w_{jT'_j},$$
	that is, $w_{-i}$ is the weight of the best allocation without agent $i$.

\end{itemize}
Since the $w_{jT}$'s do not depend on agent $i$'s report (since $i$ is in $B$), $w_{-i}$ doesn't depend on agent $i$'s report. Therefore, we can (and will) ignore this term when considering incentive compatibility below.

Note also that since the maximal partition guarantees that $w_{-i}\ge\sum_{j \in B \setminus \{i\}} w_{jT_j}$, and monotonicity of valuations in signals guarantees that $g_{-iT_i} (\tilde\bs_{-i}) \ge g_{-iT_i} (\tilde\bs_{A},\bzero_{B_{-i}})$. Therefore, the payments $p_i (\tilde\bs)$ are always nonnegative.

\begin{theorem} \label{thm:rs-vcg}
\texttt{Random-Sampling VCG} is an ex-post IC-IR mechanism that gives a 4-approximation to the optimal social welfare for any combinatorial auction setting with separable SOS valuations.
\end{theorem}

\begin{proof}
First we show that if the agents bid truthfully, then the mechanism gives a 4-approximation to social welfare. 
For every agent $i$ and bundle $T$,
\begin{eqnarray}
\E _B [w_{iT} \cdot \mathbf{1}_{i \in B}]  = \E _B [v_{iT}(\bs_{iT}, \bs_{AT}, \bzero_{B_{-i}T}) ~|~ i \in B] \cdot Pr (i \in B)
\ge \frac{ v_{iT}(\bs_{T})}{2}\cdot  \frac{1}{2},\label{eq:rs-vcg-sm}
\end{eqnarray}
where the inequality follows by applying Lemma~\ref{lem:rs-value} with $d=1$.

Let $S^*_1, \ldots, S^*_n$ be the true welfare maximizing allocation. Then,
\begin{align*}
\E_B \left[\max_ {\text{partitions }\{T_i \}}~ \sum_{i\in B} w_{iT_i}\right] &\ge
\E _B \left[\sum_i w_{iS^*_i} \cdot \mathbf{1}_{i \in B}\right]\\
&= \sum_i \E _B [w_{iS^*_i} \cdot \mathbf{1}_{i \in B}]\ge \frac{1}{4} \sum_i v_{i S^*_i}(\bs_{S_i^*}),
\end{align*}
where the last inequality follows by substituting $S_i^*$ in $T$ in Equation~\eqref{eq:rs-vcg-sm} for every $i$. Since $v_{iT}(\bs)$ is always at least $w_{iT}$, this proves the approximation ratio.

Next, we show that \texttt{RS-VCG} is universally ex-post IC. Fix a random partition $(A,B)$.
Suppose that when all agents bid truthfully
$$\{T^*_j\}_{j\in B} = \argmax_{\text{partitions }\{T_j \}}\sum_{j\in B}w_{jT_j}.$$
Suppose that
all agents but $i\in B$ bid truthfully and $i$ bids $\bs'_i$ instead of his true signal vector $\bs_i$. Let $\{T'_j\}_{j\in B}$ be the resulting allocation. Therefore, agent $i$'s utility when reporting $s'_i$ (after disregarding the $w_{-i}$ term as mentioned above) is:
\begin{eqnarray*}
v_{iT'_i} (\bs) -p_i (\bs'_i, \bs_{-i}) &= &g_{-iT'_i} (\bs_{-iT'_i}) + h_{iT'_i} (\bs_{iT'_i}) - p_i(\bs'_i, \bs_{-i})\\
& = &g_{-iT'_i} (\bs_{-iT'_i}) + h_{iT'_i} (\bs_{iT'_i}) -  \left(g_{-iT'_i} (\bs_{-iT'_i}) - g_{-iT'_i} (\bs_{AT'_i}, \bzero_{B_{-i}T'_i}) - \sum_{j \in B \setminus \{i\}} w_{jT'_j} \right)\\
&= & h_{iT'_i} (\bs_{iT'_i}) + g_{-iT'_i} (\bs_{AT'_i},\bzero_{B_{-i}T'_i}) + \sum_{j \in B \setminus \{i\}} w_{jT'_j}\\
&= &w_{iT'_i}  + \sum_{j \in B \setminus \{i\}} w_{jT'_j}  \ = \ \sum_{j\in B} w_{j T'_j}  \\
&\le&  \sum_{j\in B} w_{j T^*_j},
\end{eqnarray*}
where $\sum_{j\in B}w_{j T^*_j}$ is $i$'s utility for bidding truthfully.

Finally, we show that the mechanism is ex-post IR.
Indeed, from above, agent $i$'s utility when reporting truthfully (and without disregarding the $w_{-i}$ term) is
$$v_{iT^*_i} (\bs_{T^*_i}) -p_i (\bs)  = \sum_{j\in B} w_{j T^*_j} - w_{-i}=\sum_{j\in B} w_{j T^*_j} -
\max_ {\text{partitions }\{T'_j\} }\sum_{j \in B \setminus \{i\}} w_{jT'_j} \ge 0.$$

\end{proof}

In the case of separable $d$-SOS valuations, the \texttt{Random-Sampling VCG} is an ex-post IC-IR mechanism that gives $2(d+1)$-approximation to the social welfare. The proof is identical to Theorem \ref{thm:rs-vcg}, except that Equation~\eqref{eq:rs-vcg-sm} is changed to
\begin{eqnarray*}
	\E _B [w_{iT} \cdot \mathbf{1}_{i \in B}] \ge \frac{ v_{iT}(\bs_{T})}{2(d+1)},
\end{eqnarray*}
since we apply Lemma~\ref{lem:rs-value} with an arbitrary $d$.

\begin{remark}
Theorem \ref{thm:rs-vcg} is clearly analogous to the VCG mechanism for combinatorial auctions with private values. As with VCG for private values, in many cases, there is unlikely to be a polynomial time algorithm to compute allocations and payments. Exceptions include settings we know and love such as unit-demand auctions, additive valuations, etc.
\end{remark}

\section{Combinatorial Auctions with Single-Dimensional Signals}
\label{sec:comb-auction-single-dimension}

In this section we consider combinatorial valuations (general combinatorial auctions) with single-dimensional signals (as given by Definition~\ref{def:singlecomb}).




When the signal space of each agent is of size at most $k$, we present a mechanism that gets $(k+3)$-approximation for SOS valuations (see Section \ref{sec:k-sig-sm}), and a mechanism that gets $(2\log_2 k+4)$-approximation  for strong-SOS valuations (Definition~\ref{def:ssm}, see Section \ref{sec:k-sig-ssm} for details regarding the mechanism). For $d$-SOS and $d$-strong-SOS valuations, the mechanism generalizes to give $O(dk)$- and $O(d^2\log k)$-approximations respectively, as shown in Section~\ref{sec:d-sm}.

We first decompose the optimal welfare into two parts, $\mathsf{OTHER}$ and $\mathsf{SELF}$. Each part will be covered by a corresponding mechanism. Let $T^*=\{T_i^*\}_{i\in[n]}$ be a welfare-maximizing allocation at signal profile $\mathbf{s}$, and let $W^*(\mathbf{s})$ be the social welfare of $T^*$ at $\mathbf{s}$. Consider the following decomposition:
\begin{eqnarray}
W^*(\mathbf{s}) & = & \sum_{i}v_{iT_i^*}(\mathbf{s})\nonumber\\
& = & \sum_{i}v_{iT_i^*}(\mathbf{s}_{-i},0_i)+\sum_{i\ :\  s_i>0}\left(v_{iT_i^*}(\mathbf{s}) -v_{iT_i^*}(\mathbf{s}_{-i},0_i)\right)\nonumber\\
&\leq & \sum_{i}v_{iT_i^*}(\mathbf{s}_{-i},0_i)+\sum_{i\ :\  s_i>0}\left(v_{iT_i^*}(\mathbf{0}_{-i}, s_i) -v_{iT_i^*}(\mathbf{0})\right) \label{eq:decomp-sm} \\
& \leq & \underbrace{\sum_i v_{iT_i^*}(\mathbf{s}_{-i},0_i)}_{\text{$\mathsf{OTHER}$}}+\underbrace{\sum_{\ell=1}^{k-1}\sum_{i\ :\  s_i=\ell}v_{iT_i^*}(\mathbf{0}_{-i}, s_i)}_{\text{$\mathsf{SELF}$}},\label{eq:decomposition}
\end{eqnarray}
where Equation~\eqref{eq:decomp-sm} follows from the definition of submodularity (and therefore, also follows the definition of strong-submodularity). The last inequality follows from the non-negativity of $v_{iT_i^*}(\mathbf{0})$. The first term in the decomposition represents the contribution of others' signals to one's value from his allocated bundle, while the second term represents one's contribution to his own value. Each of these terms will be targeted using a different mechanism. Whereas the $\mathsf{OTHER}$ term will be targeted using the same mechanism in both the SOS and strong-SOS cases, the $\mathsf{SELF}$ term will be treated differently.

\subsection{$(k+3)$-approximation for SOS valuations}\label{sec:k-sig-sm}


Suppose $s_i\in \{0,1,\ldots,k-1\}$ for all $i$. The mechanism is as follows:

\vspace{0.1in}
\noindent Mechanism \texttt{$k$ signals High-Low ($k$-HL):}

\noindent With probability $p_{RT}=\frac{k-1}{k+3}$, run \texttt{Random Threshold}; otherwise, run \texttt{Random Sampling}, as described below:

\noindent Mechanism \texttt{Random Threshold}

\begin{itemize}
	\item Choose a random threshold $\ell$ uniformly in $\{1,\ldots, k-1\}$.
	\item Let $N_{\geq\ell}=\{i\ :\ s_i\geq \ell\}$ be the ``high'' agents; i.e., agents with signal at least $\ell$, and let $N_{<\ell}=[n]\setminus N_{\geq \ell}$ be the ``low'' agents.
	\item For every high agent $i\in N_{\geq\ell}$ and bundle $T$, let $\bar{v}_{iT}:=v_{iT}(\mathbf{s}_{N_{<\ell}},\mathbf{\ell}_{N_{\geq\ell}})$
    \item For every low agent $i\in N_{<\ell}$ and bundle $T$, let $\bar{v}_{iT}:=0$.
	\item Let the allocation be $$\bar{T} \in \argmax_{S=\{S_i\}_{i\in N_{\geq\ell}}} \sum_{i\in N_{\geq \ell}}\bar{v}_{iS_i}.$$ (\textit{i.e.}, the allocation that maximizes the ``welfare'' of high agents using values $\bar{v}_{iT}$.)
	\item Agent $i$ that receives bundle $\bar{T}_i$ pays $v_{i\bar{T}_i}(\mathbf{s}_{-i}, s_i=\ell-1)$.
\end{itemize}

\noindent Mechanism \texttt{Random Sampling}

\begin{itemize}
	\item Split the agents into sets $A$ and $B$ uniformly at random.
	\item For each $i\in B$ and bundle $T$, let $\tilde{v}_{iT}:=v_{ij}(\mathbf{s}_A, \mathbf{0}_B)$.
    \item For each $i\in A$ and bundle $T$, let $\tilde{v}_{iT}:=0$.
	\item Let the allocation be $$\tilde{T} \in \argmax_{S=\{S_i\}_{i\in B}} \sum_{i\in B}\tilde{v}_{iS_i}.$$ (\textit{i.e.}, the allocation that maximizes the ``welfare'' of agents in $B$ using values $\tilde{v}_{iT}$.)
    \item Charge no payments.
\end{itemize}

The \texttt{k-HL} mechanism is a random combination of two mechanisms: \texttt{Random Threshold} approximates the welfare contribution of the bidders' signals to their own value (the \textsf{SELF} term); \texttt{Random Sampling} approximates the welfare contributions of the bidders' signals to other bidders' values (the \textsf{OTHER} term).
We wish to prove the following theorem.

\begin{theorem} \label{thm:k-sig-sm}
	For every combinatorial auction setting with SOS valuations, single-dimensional signals, and signal space of size $k$, i.e. $s_i \in \{0, 1,\ldots, k-1\} \, \forall i$, mechanism \texttt{$k$-HL} is an ex-post IC-IR mechanism that gives $(k+3)$-approximation to the optimal social welfare.
\end{theorem}

We first argue that the mechanism is ex-post IC-IR.
\begin{proofof}{ex-post IC-IR}
\texttt{Random Sampling} is ex-post IC-IR since the agents that might receive items (agents in $B$) cannot change the allocation since their signals are ignored (and they pay nothing).	
	
As for \texttt{Random Threshold}, consider a threshold $\ell$ chosen by the mechanism. If the agent's signal is below $\ell$ and the agent reports $\ell$ or above, then his payment, if allocated bundle $T$ is $v_{iT}(\mathbf{s}_{-i},s_i=\ell-1)\geq v_{iT}(\mathbf{s})$; \textit{i.e.,} the agent's utility is non-positive. Bidding a different value below $\ell$ will grant the agent no items. If his value is $\ell$ or above, then bidding a different signal above $\ell$ will result in the same outcome, since the sets $N_{\geq\ell}$ and $N_{<\ell}$ remain the same. If he bids a signal below $\ell$, then he won't receive any item, and his utility will be~0, while bidding his true signal will result in non-negative utility.
\end{proofof}
In Lemma~\ref{lem:other_cover}, we prove that \texttt{Random Sampling} covers the $\mathsf{OTHER}$ component of the social welfare, and in Lemma~\ref{lem:self_cover}, we show that \texttt{Random Threshold} covers the $\mathsf{SELF}$ component.

\begin{lemma} \label{lem:self_cover}
	For SOS valuations, the \texttt{Random Threshold} mechanism gives a $(k-1)$-approximation to the \textsf{SELF} component of the optimal social welfare.
\end{lemma}
\begin{proof}
	Consider a threshold $\ell\in \{1,\ldots,k-1\}$ chosen in \texttt{Random Threshold}. Whenever $\ell$ is chosen, we have that
	\begin{eqnarray*}
		\sum_{i\ :\ s_i=\ell} \bar{v}_{iT_i^*}  \ =\  \sum_{i\ :\ s_i=\ell} v_{iT_i^*}(\mathbf{s}_{N_{<\ell}},\mathbf{\ell}_{N_{\geq\ell}})\ \geq\ \sum_{i\ :\ s_i=\ell} v_{iT_i^*}(\mathbf{0}_{-i},s_i).
	\end{eqnarray*}
	Since \texttt{Random Threshold} chooses an allocation $\bar{T}=\{\bar{T}_i\}_{i\in N_{\geq \ell}}$ that maximizes the welfare under $\bar{v}_{iT}$'s, the value of the allocation is only larger than the left expression above. Because $v_{i\bar{T}_i}(\mathbf{s})\geq \bar{v}_{i\bar{T}_i}$, we get that if $\ell$ was chosen, which happens with probability $\frac{1}{k-1}$, the welfare achieved is at least $\sum\limits_{i\ :\ s_i=\ell} v_{iT_i^*}(\mathbf{0}_{-i},s_i).$ Therefore, the welfare from running \texttt{Random Threshold} is at least
	$$\sum_{\ell=1}^{k-1} \frac{1}{k-1}\sum_{i\ :\ s_i=\ell} v_{iT_i^*}(\mathbf{0}_{-i},s_i),\geq \frac{\mathsf{SELF}}{k-1}.$$
	
\end{proof}

\begin{lemma} \label{lem:other_cover}
	For SOS valuations, the \texttt{Random Sampling} mechanism gives a $4$-approximation to the \textsf{OTHER} component of the optimal social welfare.
\end{lemma}
\begin{proof}
	Consider a set $T$. Using an application of the Key Lemma~\ref{lem:rs-value} with respect to $v_{iT}(\mathbf{s}_{-i}, 0_i)$, we see that
	\begin{eqnarray}
		\E_{A,B}[\tilde{v}_{iT}]\geq \Pr[i\in B]\cdot \E_{A,B}[\tilde{v}_{iT}\ |\ i\in B]=\frac{1}{2}\E_{A,B \setminus i}[\tilde{v}_{iT}\ |\ i\in B]\geq\frac{1}{4}v_{iT}(\mathbf{s}_{-i},0_i).\label{eq:other_cover}
	\end{eqnarray}
	
	Therefore, the expected weight of the allocation $\{T_i^*\}_{i\in [n]}$ using weights $\tilde{v}_{iT}$'s is
	$$\E_{A,B}\bigg[\sum_{i}\tilde{v}_{iT_i^*}\bigg] \ =\  \sum_{i}\E_{A,B}\bigg[\tilde{v}_{iT_i^*}\bigg]\ \geq\ \sum_{i} \frac{1}{4}v_{iT_i^*}(\mathbf{s}_{-i},0_i)\ =\ \frac{\mathsf{OTHER}}{4}.$$
	Since the mechanism chooses the optimal allocation according to the $\tilde{v}_{iT}$'s, its weight can only be larger. Moreover, since $\tilde{v}_{iT} = v_{iT}(\mathbf{s}_{-i},0)\leq v_{iT}(\mathbf{s})$, the welfare achieved by the mechanism is at least $\frac{\mathsf{OTHER}}{4}$, as desired.
\end{proof}

We conclude by proving the claimed approximation ratio.

\begin{proofof}{approximation}
	 According to Lemma~\ref{lem:self_cover}, \texttt{Random Threshold} approximates $\mathsf{SELF}$ to a factor of $k-1$. According to Lemma~\ref{lem:other_cover} that \texttt{Random Sampling} approximates $\mathsf{OTHER}$ to a factor of $4$.
	Therefore, running \texttt{Random Threshold} with probability $p_{RT}$ and \texttt{Random Sampling} with probability $1-p_{RT}$ yields a welfare of
	\begin{eqnarray*}
		p_{RT}\frac{\mathsf{SELF}}{k-1} + (1-p_{RT})\frac{\mathsf{OTHER}}{4} &= & \frac{k-1}{k+3}\cdot \frac{\mathsf{SELF}}{k-1} + \frac{4}{k+3}\cdot \frac{\mathsf{OTHER}}{4}\\
		& =  & \frac{\mathsf{SELF}+\mathsf{OTHER}}{k+3}
		\  \geq \ \frac{W^*(\bs)}{k+3},
	\end{eqnarray*}
	where the inequality follows Equation~\eqref{eq:decomposition}.
\end{proofof}

\subsection{$O(\log k)$-Approximation with Strong-SOS Valuations} \label{sec:k-sig-ssm}

Strong-SOS valuations means the effect on the valuation is concave in one's own signal.
This allows us to use a bucketing technique in order to give an $O(\log k)$-approximation 
to the \textsf{SELF} component in the decomposition depicted by Equation~\eqref{eq:decomposition}.

Consider the \textsf{SELF} term in Equation~\eqref{eq:decomposition}. We can bound this term as follows:
\begin{eqnarray}
\mathsf{SELF} &=&	\sum_{\ell=1}^{k-1}\sum_{i\ :\  s_i=\ell}v_{iT_i^*}(\mathbf{0}_{-i}, s_i)\nonumber\\
& = & \sum_{\ell=1}^{\log_2 k}\sum_{i\ :\ 2^{\ell-1}\leq s_i < 2^\ell} v_{iT_i^*}(\mathbf{0}_{-i}, s_i)\nonumber\\
& \leq & \sum_{\ell=1}^{\log_2 k}\sum_{i\ :\ 2^{\ell-1}\leq s_i < 2^\ell} v_{iT_i^*}(\mathbf{0}_{-i}, {2^{\ell-1}}_i), \label{eq:selfbound_strong_submod}
\end{eqnarray}
where the inequality follows the definition of strong-SOS valuations.

We introduce mechanism \texttt{Random Bucket} to give an $O(\log k)$-approximation
to the upper bound in Equation~\eqref{eq:selfbound_strong_submod}.

\vspace{0.1in}
\noindent Mechanism \texttt{Random Bucket}:
\begin{itemize}
	\item choose $\ell$ uniformly in $\{1,\ldots, \log_2 k\}$.
	\item Let $N_{B_\ell}=\{i\ :\ \mbox{such that }s_i\geq 2^{\ell-1}\}$ be the agents with signal at least $2^{\ell-1}$ and $N_{\neg B_\ell}=[n]\setminus N_{B_\ell}$.
	\item For $i\in N_{B_\ell}$ and bundle $T$, let $\bar{v}_{iT}:=v_{iT}(\mathbf{s}_{N_{\neg B_\ell}},\mathbf{2^{\ell-1}}_{N_{B_\ell}})$ (and $\bar{v}_{iT}:=0$ for $i\in N_{\neg B_\ell}$).
	\item Let the allocation be $$\bar{T} \in \argmax_{S=\{S_i\}_{i\in N_{B_\ell}}} \sum_{i\in N_{B_\ell}}\bar{v}_{iS_i}.$$ (\textit{i.e.}, the allocation that maximizes the ``welfare'' of high agents using values $\bar{v}_{iT}$.)
	\item Agent $i$ that receives bundle $\bar{T}_i$ pays $v_{i\bar{T}_i}(\mathbf{s}_{-i}, s_i=2^{\ell-1}-1)$.
\end{itemize}

We show the following approximation guarantee regarding \texttt{Random Bucket}.

\begin{lemma} \label{lem:self_cover_strongly_submod}
	For strong-SOS valuations, the \texttt{Random Bucket} mechanism is ex-post IC-IR and gives a $2\log_2{k}$ approximation to the \textsf{SELF} component of the optimal social welfare.
\end{lemma}	
\begin{proof}
	The proof of ex-post IC-IR is identical to that of mechanism \texttt{Random Threshold}, as both are threshold-based mechanisms. The proof of the approximation guarantee is also very similar to that of \texttt{Random Threshold}.
	
	Consider a threshold $2^{\ell-1}$ for $\ell\in \{1,\ldots,k-1\}$ chosen in \texttt{Random Bucket}. Whenever $\ell$ is chosen, we have that
	\begin{eqnarray*}
		\sum_{i\ :\ 2^{\ell-1}\leq s_i < 2^\ell} \bar{v}_{iT_i^*}  \ = \  \sum_{i\ :\ 2^{\ell-1}\leq s_i < 2^\ell} v_{iT_i^*}(\mathbf{s}_{N_{\neg B_\ell}}	,\mathbf{2^{\ell-1}}_{N_{B_\ell}}) \ \geq\ \sum_{i\ :\ 2^{\ell-1}\leq s_i < 2^\ell} v_{iT_i^*}(\mathbf{0}_{-i},{2^{\ell-1}}_i).
	\end{eqnarray*}
	Since \texttt{Random Bucket} chooses an allocation that maximizes the $\bar{v}_{iT}$'s, the value of the allocation is only larger. Because $v_{i\bar{T}_i}(\mathbf{s})\geq \bar{v}_{i\bar{T}_i}$, we get that if $\ell$ was chosen, which happens with probability $\frac{1}{\log_2{k}}$, the welfare achieved is at least $\sum\limits_{i\ :\ 2^{\ell-1}\leq s_i < 2^\ell} v_{iT_i^*}(\mathbf{0}_{-i},s_i).$ Therefore, the welfare from running \texttt{Random Bucket} is at least
	$$\sum_{\ell=1}^{\log_2{k}} \frac{1}{\log_2{k}}\sum_{i\ :\ 2^{\ell-1}\leq s_i < 2^\ell} v_{iT_i^*}(\mathbf{0}_{-i},s_i),\geq \frac{\mathsf{SELF}}{2\log_2{k}}.$$
	
\end{proof}

Mechanism \texttt{$k$-signals Strong-SOS} (\texttt{$k$-SS}) runs \texttt{Random Bucket} with probability $p_{RB}=\frac{log_2{k}}{log_2{k}+2}$ 
and mechanism \texttt{Random Sampling} with probability $1-p_{RB}$.

\begin{theorem}	\label{thm:k-sig-ssm}
	For every combinatorial auction with single-dimensional signals with strong-SOS valuations and signal space of size $k$, i.e. $s_i \in \{0, 1,\ldots, k-1\} \, \forall i$, mechanism \texttt{$k$-SS} is ex-post IC-IR, and gives $(2\log_2{k}+4)$-approximation to the optimal social welfare.
\end{theorem}
\begin{proof}
	We already established that both \texttt{Random Bucket} and \texttt{Random Sampling} are ex-post IC-IR, hence  \texttt{$k$-SS} is ex-post IC-IR as well. As for the approximation, according to Lemma~\ref{lem:self_cover_strongly_submod}, with probability $p_{RB}$ we get $2\log_2{k}$-approximation 
	to \textsf{SELF}, and according to Lemma~\ref{lem:other_cover}, with probability $1-p_{RB}$ we get a $4$-approximation 
	to \textsf{OTHER}. Overall, the expected welfare is at least
	\begin{eqnarray*}
		p_{RB}\frac{\mathsf{SELF}}{2\log_2{k}}+(1-p_{RB})\frac{\mathsf{OTHER}}{4} & = & \frac{\mathsf{SELF} + \mathsf{OTHER}}{2\log_2{k}+4}\\
		& \geq & \frac{W^*}{2\log_2{k}+4},
	\end{eqnarray*}
	as desired.
	
\end{proof}

\section{Open Problems}
\label{sec:open-problems}

Our analysis and results suggest many open problems:
\begin{itemize}
\item For combinatorial auctions with multi-dimensional signals: is separability a necessary condition for achieving constant approximation to welfare?
This problem is open even for single-dimensional signals, and even for ``simple" combinatorial valuations, such as unit-demand.
\item For single-parameter SOS valuations, downward closed feasibility, and single-dimensional signals, closing the gap between $1/4$ and $1/2$ is open.
\item The exact same gap applies for combinatorial, separable-SOS valuations with multi-dimensional signals.
\item How does the distinction between SOS and strong-SOS affect the problems above, if at all?
\item When considering the relaxation of SOS valuations to $d$-SOS valuations, there is a gap between the positive and negative results
with respect to the dependence on $d$.
\end{itemize}
More generally, what other classes of valuations give rise to approximately efficient mechanisms in settings with interdependent valuations?

\vspace{0.1in}
\par\noindent{\bf Acknowledgements} 
We gratefully thank an anonymous referee who pointed out that many of the proofs in this paper, hold, with minor adjustments, for subadditive over signals valuations. Surprisingly, submodular over signals valuations are not a special case of subadditive over signals valuations. However, strong submodular over signals valuations are so. The actual situation is rather subtle and we will address this issue in a subsequent version of this paper.

\bibliographystyle{plainnat}
\bibliography{main}

\begin{thebibliography}{44}
\providecommand{\natexlab}[1]{#1}
\providecommand{\url}[1]{\texttt{#1}}
\expandafter\ifx\csname urlstyle\endcsname\relax
  \providecommand{\doi}[1]{doi: #1}\else
  \providecommand{\doi}{doi: \begingroup \urlstyle{rm}\Url}\fi

\bibitem[Abraham et~al.(2011)Abraham, Athey, Babaioff, and
  Grubb]{abraham2011peaches}
Ittai Abraham, Susan Athey, Moshe Babaioff, and M~Grubb.
\newblock Peaches.
\newblock \emph{Lemons, and Cookies: Designing Auction Markets with Dispersed
  Information}, 2011.

\bibitem[Athey(2001)]{athey01}
Susan Athey.
\newblock Single crossing properties and the existence of pure strategy
  equilibria in games of incomplete information.
\newblock \emph{Econometrica}, 69\penalty0 (4):\penalty0 861--889, 2001.
\newblock ISSN 00129682, 14680262.

\bibitem[Ausubel(1999)]{ausubel1999generalized}
Lawrence~M Ausubel.
\newblock A generalized vickrey auction.
\newblock \emph{Econometrica}, 1999.

\bibitem[Babaioff et~al.(2012)Babaioff, Kleinberg, and Paes~Leme]{BabaioffKL12}
Moshe Babaioff, Robert Kleinberg, and Renato Paes~Leme.
\newblock Optimal mechanisms for selling information.
\newblock In \emph{Proceedings of the 13th ACM Conference on Electronic
  Commerce}, EC '12, pages 92--109, New York, NY, USA, 2012. ACM.
\newblock ISBN 978-1-4503-1415-2.

\bibitem[Bergemann and Morris(2005)]{Bergemann05}
Dirk Bergemann and Stephen Morris.
\newblock Robust mechanism design.
\newblock \emph{Econometrica}, 73\penalty0 (6):\penalty0 1771--1813, 2005.

\bibitem[Bergemann et~al.(2009)Bergemann, Shi, and
  V{\"a}lim{\"a}ki]{bergemann2009information}
Dirk Bergemann, Xianwen Shi, and Juuso V{\"a}lim{\"a}ki.
\newblock Information acquisition in interdependent value auctions.
\newblock \emph{Journal of the European Economic Association}, 7\penalty0
  (1):\penalty0 61--89, 2009.

\bibitem[Bian et~al.(2017)Bian, Levy, Krause, and Buhmann]{Bian17}
Andrew~An Bian, Kfir~Yehuda Levy, Andreas Krause, and Joachim~M. Buhmann.
\newblock Non-monotone continuous dr-submodular maximization: Structure and
  algorithms.
\newblock In \emph{Advances in Neural Information Processing Systems 30: Annual
  Conference on Neural Information Processing Systems 2017, 4-9 December 2017,
  Long Beach, CA, {USA}}, pages 486--496, 2017.

\bibitem[Bikhchandani(2006)]{Bikhchandani2006}
Sushil Bikhchandani.
\newblock Ex post implementation in environments with private goods.
\newblock 2006.

\bibitem[Chawla et~al.(2014)Chawla, Fu, and Karlin]{CFK}
Shuchi Chawla, Hu~Fu, and Anna Karlin.
\newblock Approximate revenue maximization in interdependent value settings.
\newblock In \emph{Proceedings of the Fifteenth ACM Conference on Economics and
  Computation}, pages 277--294, New York, NY, USA, 2014. ACM.
\newblock ISBN 978-1-4503-2565-3.

\bibitem[Che et~al.(2015)Che, Kim, and Kojima]{che2015efficient}
Yeon-Koo Che, Jinwoo Kim, and Fuhito Kojima.
\newblock Efficient assignment with interdependent values.
\newblock \emph{Journal of Economic Theory}, 158:\penalty0 54--86, 2015.

\bibitem[Clarke(1971)]{clarke}
Edward~H Clarke.
\newblock Multipart pricing of public goods.
\newblock \emph{Public choice}, 11\penalty0 (1):\penalty0 17--33, 1971.

\bibitem[Constantin and Parkes(2007)]{ConstantinP07}
Florin Constantin and David~C. Parkes.
\newblock On revenue-optimal dynamic auctions for bidders with interdependent
  values.
\newblock In \emph{{AMEC/TADA}}, volume~13 of \emph{Lecture Notes in Business
  Information Processing}, pages 1--15. Springer, 2007.

\bibitem[Constantin et~al.(2007)Constantin, Ito, and Parkes]{ConstantinIP07}
Florin Constantin, Takayuki Ito, and David~C. Parkes.
\newblock Online auctions for bidders with interdependent values.
\newblock In \emph{{AAMAS}}, page 110. {IFAAMAS}, 2007.

\bibitem[Cr\'{e}mer and McLean(1985)]{CremerMcLean85}
Jacques Cr\'{e}mer and Richard~P. McLean.
\newblock Optimal selling strategies under uncertainty for a discriminating
  monopolist when demands are interdependent.
\newblock \emph{Econometrica}, 53\penalty0 (2):\penalty0 345--361, 1985.
\newblock ISSN 00129682, 14680262.

\bibitem[Cr\'{e}mer and McLean(1988)]{CremerMcLean88}
Jacques Cr\'{e}mer and Richard~P. McLean.
\newblock Full extraction of the surplus in bayesian and dominant strategy
  auctions.
\newblock \emph{Econometrica}, 56\penalty0 (6):\penalty0 1247--1257, 1988.
\newblock ISSN 00129682, 14680262.

\bibitem[Dasgupta and Maskin(2000)]{dasgupta2000efficient}
Partha Dasgupta and Eric Maskin.
\newblock Efficient auctions.
\newblock \emph{The Quarterly Journal of Economics}, 115\penalty0 (2):\penalty0
  341--388, 2000.

\bibitem[d'Aspremont and G\'erard-Varet(1982)]{Aspremont82}
C~d'Aspremont and L.-A G\'erard-Varet.
\newblock Bayesian incentive compatible beliefs.
\newblock \emph{Journal of Mathematical Economics}, 10\penalty0 (1):\penalty0
  83 -- 103, 1982.
\newblock ISSN 0304-4068.

\bibitem[Dobzinski et~al.(2011)Dobzinski, Fu, and Kleinberg]{DobzinskiFK11}
Shahar Dobzinski, Hu~Fu, and Robert~D. Kleinberg.
\newblock Optimal auctions with correlated bidders are easy.
\newblock In \emph{Proceedings of the Forty-third Annual ACM Symposium on
  Theory of Computing}, pages 129--138, New York, NY, USA, 2011. ACM.
\newblock ISBN 978-1-4503-0691-1.

\bibitem[Eden et~al.(2018)Eden, Feldman, Fiat, and Goldner]{eden18}
Alon Eden, Michal Feldman, Amos Fiat, and Kira Goldner.
\newblock Interdependent values without single-crossing.
\newblock In \emph{Proceedings of the 2018 ACM Conference on Economics and
  Computation}, EC '18, pages 369--369, New York, NY, USA, 2018. ACM.
\newblock ISBN 978-1-4503-5829-3.

\bibitem[Goldberg et~al.(2001)Goldberg, Hartline, and Wright]{GHW01}
Andrew~V. Goldberg, Jason~D. Hartline, and Andrew Wright.
\newblock Competitive auctions and digital goods.
\newblock In \emph{Proceedings of the 12th Annual ACM-SIAM Symposium on
  Discrete Algorithms}, pages 735--744, Philadelphia, PA, USA, 2001.
\newblock ISBN 0-89871-490-7.

\bibitem[Groves(1973)]{groves}
Theodore Groves.
\newblock Incentives in teams.
\newblock \emph{Econometrica: Journal of the Econometric Society}, pages
  617--631, 1973.

\bibitem[Ito and Parkes(2006)]{ItoP06}
Takayuki Ito and David~C. Parkes.
\newblock Instantiating the contingent bids model of truthful interdependent
  value auctions.
\newblock In \emph{{AAMAS}}, pages 1151--1158. {ACM}, 2006.

\bibitem[Jehiel and Moldovanu(2001)]{jehiel2001efficient}
Philippe Jehiel and Benny Moldovanu.
\newblock Efficient design with interdependent valuations.
\newblock \emph{Econometrica}, 69\penalty0 (5):\penalty0 1237--1259, 2001.

\bibitem[Jehiel et~al.(2006)Jehiel, Meyer-ter Vehn, Moldovanu, and
  Zame]{jehiel2006limits}
Philippe Jehiel, Moritz Meyer-ter Vehn, Benny Moldovanu, and William~R Zame.
\newblock The limits of ex post implementation.
\newblock \emph{Econometrica}, 74\penalty0 (3):\penalty0 585--610, 2006.

\bibitem[Kempe et~al.(2013)Kempe, Syrgkanis, and Tardos]{kempe2013information}
David Kempe, Vasilis Syrgkanis, and Eva Tardos.
\newblock Information asymmetries in common-value auctions with discrete
  signals.
\newblock \emph{SSRN eLibrary}, 2013.

\bibitem[Klein et~al.(2008)Klein, Moreno, Parkes, Plakosh, Seuken, and
  Wallnau]{KleinMPPSW08}
Mark Klein, Gabriel~A. Moreno, David~C. Parkes, Daniel Plakosh, Sven Seuken,
  and Kurt~C. Wallnau.
\newblock Handling interdependent values in an auction mechanism for bandwidth
  allocation in tactical data networks.
\newblock In \emph{NetEcon}, pages 73--78. {ACM}, 2008.

\bibitem[Klemperer(1998)]{klemperer1998auctions}
Paul Klemperer.
\newblock Auctions with almost common values: The wallet game'and its
  applications.
\newblock \emph{European Economic Review}, 42\penalty0 (3):\penalty0 757--769,
  1998.

\bibitem[Krishna(2009)]{krishna2009auction}
Vijay Krishna.
\newblock \emph{Auction theory}.
\newblock Academic press, 2009.

\bibitem[Li(2016)]{li2016approximation}
Yunan Li.
\newblock Approximation in mechanism design with interdependent values.
\newblock \emph{Games and Economic Behavior}, 2016.

\bibitem[Maskin(1992)]{maskin1992}
Eric Maskin.
\newblock Auctions and privatization.
\newblock \emph{Privatization, H. Siebert, ed. (Institut fur Weltwirtschaften
  der Universita¨t Kiel: 1992)}, pages 115--–136, 1992.

\bibitem[McLean and Postlewaite(2015)]{McLean2015}
Richard McLean and Andrew Postlewaite.
\newblock Implementation with interdependent valuations.
\newblock \emph{Theoretical Economics}, 10\penalty0 (3), 2015.

\bibitem[Milgrom and Weber(1982)]{milgrom1982theory}
Paul~R Milgrom and Robert~J Weber.
\newblock A theory of auctions and competitive bidding.
\newblock \emph{Econometrica: Journal of the Econometric Society}, pages
  1089--1122, 1982.

\bibitem[Milgrom(2004)]{milgrom2004putting}
Paul~Robert Milgrom.
\newblock \emph{Putting auction theory to work}.
\newblock Cambridge University Press, 2004.

\bibitem[Myerson(1981)]{myerson1981optimal}
Roger~B Myerson.
\newblock Optimal auction design.
\newblock \emph{Mathematics of operations research}, 6\penalty0 (1):\penalty0
  58--73, 1981.

\bibitem[Niazadeh et~al.(2018)Niazadeh, Roughgarden, and Wang]{Niazadeh18}
Rad Niazadeh, Tim Roughgarden, and Joshua~R. Wang.
\newblock Optimal algorithms for continuous non-monotone submodular and
  dr-submodular maximization.
\newblock In \emph{Annual Conference on Neural Information Processing Systems
  2018, 3-8 December 2018, Montr{\'{e}}al, Canada.}, pages 9617--9627, 2018.

\bibitem[Papadimitriou and Pierrakos(2011)]{PapadimitriouPierrakos10}
Christos~H. Papadimitriou and George Pierrakos.
\newblock On optimal single-item auctions.
\newblock In \emph{Proceedings of the 43rd {ACM} Symposium on Theory of
  Computing, {STOC} 2011, San Jose, CA, USA, 6-8 June 2011}, pages 119--128,
  2011.

\bibitem[Robu et~al.(2013)Robu, Parkes, Ito, and Jennings]{RobuPIJ13}
Valentin Robu, David~C. Parkes, Takayuki Ito, and Nicholas~R. Jennings.
\newblock Efficient interdependent value combinatorial auctions with single
  minded bidders.
\newblock In \emph{{IJCAI}}, pages 339--345. {IJCAI/AAAI}, 2013.

\bibitem[Rochet(1987)]{rochet1987necessary}
Jean-Charles Rochet.
\newblock A necessary and sufficient condition for rationalizability in a
  quasi-linear context.
\newblock \emph{Journal of mathematical Economics}, 16\penalty0 (2):\penalty0
  191--200, 1987.

\bibitem[Ronen(2001)]{ronen2001approximating}
Amir Ronen.
\newblock On approximating optimal auctions.
\newblock In \emph{Proceedings of the 3rd ACM conference on Electronic
  Commerce}, pages 11--17. ACM, 2001.

\bibitem[Roughgarden and Talgam-Cohen(2016)]{RTCoptimalrev}
Tim Roughgarden and Inbal Talgam-Cohen.
\newblock Optimal and robust mechanism design with interdependent values.
\newblock \emph{ACM Trans. Econ. Comput.}, 4\penalty0 (3):\penalty0
  18:1--18:34, June 2016.
\newblock ISSN 2167-8375.

\bibitem[Soma and Yoshida(2015)]{Soma_2015}
Tasuku Soma and Yuichi Yoshida.
\newblock A generalization of submodular cover via the diminishing return
  property on the integer lattice.
\newblock In \emph{Advances in Neural Information Processing Systems 28}, pages
  847--855. Curran Associates, Inc., 2015.

\bibitem[Vickrey(1961)]{vickrey1961counterspeculation}
William Vickrey.
\newblock Counterspeculation, auctions, and competitive sealed tenders.
\newblock \emph{The Journal of finance}, 16\penalty0 (1):\penalty0 8--37, 1961.

\bibitem[Vohra(2007)]{vohra2007paths}
Rakesh Vohra.
\newblock Paths, cycles and mechanism design.
\newblock \emph{Preprint}, 2007.

\bibitem[Wilson(1969)]{wilson1969communications}
Robert~B Wilson.
\newblock {Competitive Bidding with Disparate Information}.
\newblock \emph{Management Science}, 15\penalty0 (7):\penalty0 446--452, 1969.

\end{thebibliography}

\medskip

\appendix

\section{Unit-Demand Valuations with Single-Crossing}
\label{sec:ud-lb}
Whereas single-crossing is a strong enough condition to implement the fully efficient mechanism in a variety of single-parameter environments, generalizations of this condition fail even in the simplest multi-parameter environments. We consider the case where bidders are unit demand and each bidder has a scalar as a signal. We define single-crossing for this setting as follows.

\begin{definition}[Single-crossing for unit-demand valuations]
	A valuation profile $\mathbf{v}$ is said to be single crossing if for every agent $i$, signals $\mathbf{s}_{-i}$, item $j$ and agent $\ell$,
	\begin{eqnarray}
		\frac{\partial}{\partial s_i} v_{ij}(\mathbf{s}_{-i}, s_i) \geq \frac{\partial}{\partial s_i} v_{\ell j}(\mathbf{s}_{-i}, s_i). \label{eq:ud-sc}
	\end{eqnarray}
\end{definition}

In this section, we show that in the case two non-identical items are for sale, and the valuations are unit demand and satisfy single-crossing as defined in Equation~\eqref{eq:ud-sc}, any truthful mechanism is bounded away from achieving full efficiency.

In order to give the lower bound, we first give a characterization of ex-post IC and IR mechanisms in multi-dimensional environments in interdependent values settings (Section~\ref{sec:cycle-mon}). We then turn to prove the lower bound (Section~\ref{sec:ud-lb-proof}). 

\subsection{Cycle Monotonicity}
\label{sec:cycle-mon}
In the IPV model, \cite{rochet1987necessary} introduced cycle monotonicity as a necessary and sufficient condition on the allocation to be implementable in dominant strategies (DSIC) for multidimensional environments. It was noticed that a straightforward analogue holds for the IDV value model, for ex-post implementability (EPIC) (in  \cite{vohra2007paths}, this fact is stated without a proof).

Fix a feasible allocation rule $\mathbf{x}=\{x_i\}_{i\in[n]}$, where $x_{iT}(\mathbf{s})$ is the probability agent $i$ receives a bundle $T$ under bid profile $\mathbf{s}$. For each agent $i$, consider the graph $G^{\mathbf{x}}_i$ where there is a vertex for each signal profile $\mathbf{s}$, and there is a directed edge from $\mathbf{s}$ to $\mathbf{t}$ if $\mathbf{s}_{-i}=\mathbf{t}_{-i}$. The weight of edge $(\mathbf{s},\mathbf{t})$ is
$$
w(\mathbf{s},\mathbf{t}) = \E_{T\sim x_i(\mathbf{s})}[v_{iT}(\mathbf{s})]-\E_{T\sim x_i(\mathbf{t})}[v_{iT}(\mathbf{s})] =
\sum_{T \subseteq [m]} x_{iT}(\mathbf{s})v_{iT}(\mathbf{s}) - \sum_{T \subseteq [m]} x_{iT}(\mathbf{t})v_{iT}(\mathbf{s}).
$$ 

The following theorem states that a necessary and sufficient condition for ex-post implementability of $\mathbf{x}$ is that for every agent $i$, every directed cycle in $G^{\mathbf{x}}_i$ is non-negative.
The proof is a straightforward adjustment of the original proof in \cite{rochet1987necessary}, and is given below  
for completeness.

\begin{theorem}\label{thm:cycle-mon}
	The allocation rule $\mathbf{x}$ is implementable by an ex-post IC mechanism if and only if for every agent $i$, all directed cycles in  $G^{\mathbf{x}}_i$ have non-negative weight.
\end{theorem}

\begin{proof}
We first show that if the allocation rule is implementable, then there are no negative cycles. Fix some payment rule $\mathbf{p}=\{p_i\}_{i\in [n]}$, where $p_i(\mathbf{s})$ is the payment of agent $i$ under bid profile $\mathbf{s}$. 
Let $\mathbf{s}_{-i}$ be the real signals of all bidders except $i$, and consider a cycle $\mathbf{s}^1\rightarrow \mathbf{s}^2\rightarrow\ldots\rightarrow \mathbf{s}^\ell\rightarrow\mathbf{s}^1$ in $G^{\mathbf{x}}_i$, where $\mathbf{s}^t=(\mathbf{s}_{-i},s_i=\zeta_t)$ for $t\in [\ell]$. 
Since $(\mathbf{x},\mathbf{p})$ is an ex-post IC mechanism, for every true signal $s_i=s$, agent $i$ is at least as well off bidding $s$ than any other bid $s'$. We get that
	\begin{eqnarray*}
		\E_{T\sim x_i(\mathbf{s}^1)}[v_{iT}(\mathbf{s}^1)]-p_i(\mathbf{s}^1)&\geq &\E_{T\sim x_i(\mathbf{s}^2)}[v_{iT}(\mathbf{s}^1)]-p_i(\mathbf{s}^2)\\
		&\vdots&\\
		\E_{T\sim x_i(\mathbf{s}^{\ell-1})}[v_{iT}(\mathbf{s}^{\ell-1})]-p_i(\mathbf{s}^{\ell-1})&\geq &\E_{T\sim x_i(\mathbf{s}^\ell)}[v_{iT}(\mathbf{s}^{\ell-1})]-p_i(\mathbf{s}^\ell)\\
		\E_{T\sim x_i(\mathbf{s}^{\ell})}[v_{iT}(\mathbf{s}^{\ell})]-p_i(\mathbf{s}^{\ell})&\geq &\E_{T\sim x_i(\mathbf{s}^1)}[v_{iT}(\mathbf{s}^{\ell})]-p_i(\mathbf{s}^1)
	\end{eqnarray*}

	Summing over the above inequalities and using the convention that $\ell+1=1$, we get that
	\begin{eqnarray*}
		& &\sum_{j=1}^\ell \E_{T\sim x_i(\mathbf{s}^j)}[v_{iT}(\mathbf{s}^j)] - \sum_{j=1}^\ell p_i(\mathbf{s}^j) \geq \sum_{j=1}^\ell \E_{T\sim x_i(\mathbf{s}^{j+1})}[v_{iT}(\mathbf{s}^j)] - \sum_{j=1}^\ell p_i(\mathbf{s}^j)\\
		& \iff& \sum_{j=1}^\ell \left(\E_{T\sim x_i(\mathbf{s}^j)}[v_{iT}(\mathbf{s}^j)] - \E_{T\sim x_i(\mathbf{s}^{j+1})}[v_{iT}(\mathbf{s}^j)]\right)\geq 0,
	\end{eqnarray*}
	where the LHS of the last inequality is exactly the weight of the cycle.

	We now show how to compute payments that implement a given allocation rule $\mathbf{x}$ that induces no negative cycles for any $i$ and $G_i^{\mathbf{x}}$. Given $G_i^{\mathbf{x}}$, one can compute payments as follows.
	\begin{itemize}
		\item Add a dummy node $d$ with edges of weight~0 to all nodes in $G_i^{\mathbf{x}}$.
		\item For every node $\mathbf{s}$ of $G_i^{\mathbf{x}}$, let $\delta(\mathbf{s})$ be the distance of the shortest path from $d$ to $\mathbf{s}$.
		\item Set $p_i(\mathbf{s})= -\delta(\mathbf{s})$.
	\end{itemize}
	Fix signals of the other players $\mathbf{s}_{-i}$. Let $s$ be player $i$'s true signal and $s'$ be some other possible signal for $i$. Denote $\mathbf{s}=(\mathbf{s}_{-i},s)$ and $\mathbf{s'}=(\mathbf{s}_{-i},s')$. Consider the nodes $\mathbf{s}$ and $\mathbf{s'}$ in $G_i^{\mathbf{x}}$. Since $\delta(\mathbf{s'})$ is the length of the shortest path from $d$, it must be that $$\delta(\mathbf{s'}) \leq \delta(\mathbf{s}) + w(\mathbf{s},\mathbf{s'}),$$ where $w(\mathbf{s},\mathbf{s'})$ is the weight of the edge from $\mathbf{s}$ to $\mathbf{s'}$. Substituting $w(\mathbf{s},\mathbf{s'}) = \E_{T\sim x_i(\mathbf{s})}[v_{iT}(\mathbf{s})]-\E_{T\sim x_i(\mathbf{s'})}[v_{iT}(\mathbf{s})]$, $p_i(\mathbf{s})=-\delta(\mathbf{s})$, and $p_i(\mathbf{s'})=-\delta(\mathbf{s'})$, we get
	$$\E_{T\sim x_i(\mathbf{s})}[v_{iT}(\mathbf{s})]-p_i(\mathbf{s})\geq \E_{T\sim x_i(\mathbf{s'})}[v_{iT}(\mathbf{s})]- p_i(\mathbf{s'}),$$
	 as desired.	
\end{proof} 

\subsection{Lower Bounds for Deterministic and Randomized Mechanisms}\label{sec:ud-lb-proof}

\begin{lemma} \label{lem:sclowertwodet}
There exists a setting with two items and two agents with unit-demand and single crossing valuations, such that no deterministic truthful mechanism achieves more than $1/2$ of the optimal welfare.
\end{lemma}

\begin{figure}[h!]
\centering
\includegraphics[scale=.5]{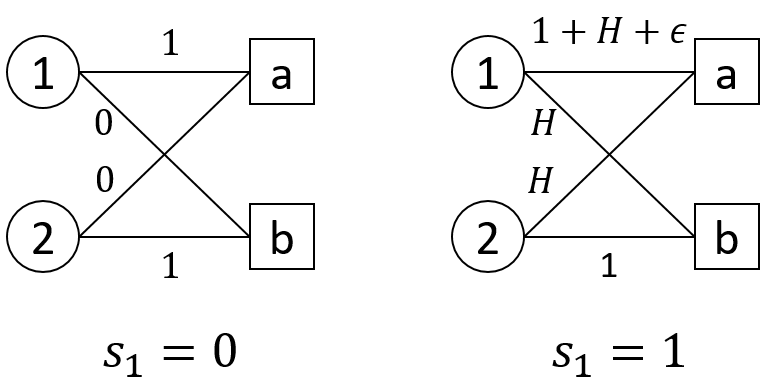}
\caption{An instance with unit-demand single-crossing valuations where no deterministic truthful allocation achieves more than a half of the optimal welfare.}
\centering
\label{fig:2-lb}
\end{figure}

\begin{proof}
	Consider the setting depicted in Figure \ref{fig:2-lb}, with two agents, 1 and~2, and two items, $a$ and $b$. $s_1\in\{0,1\}$ and $s_2$ is fixed.
	The values at $s_1=0$ are $$v_{1a}(0)=1, v_{1b}(0)=0, v_{2a}(0)=0, v_{1b}(0)=1,$$ and at $s_1=1$ are
	$$v_{1a}(1)=1+H+\epsilon, v_{1b}(1)=H, v_{2a}(1)=H, v_{1b}(1)=1,$$
	for some arbitrarily large $H$ and a sufficiently small $\epsilon$.
	One can easily verify that the valuations satisfy Equation~\eqref{eq:ud-sc}, and hence single crossing; indeed, when agent~1's signal increases, the valuation of agent~1 for each one of the item increases by more than the change in agent 2's valuation.
	
We show that no deterministic truthful mechanism can get better than 2-approximation. In order to get better than 2-approximation, the mechanism must allocate item $a$ to agent~1 and item $b$ to bidder~2 at signal $s_1=0$. At $s_1=1$, allocating item $b$ to agent~1 and item $a$ to agent~2 obtains a welfare of $2H$, while any other allocation obtains at most a welfare of $H+2+\epsilon$. Since $H$ can be arbitrarily large, one must allocate item $b$ to agent~1 and item $a$ to agent~2 at signal $s_1=1$ in order to get an approximation ratio better than~2. Consider such an allocation rule $\mathbf{x}$, and the graph $G_1^\mathbf{x}$. This graph has one cycle, with one edge from $s_1=0$ to $s_1=1$ and one edge from $s_1=1$ to $s_1=0$. The weight of this cycle is $$(v_{1a}(0)-v_{1b}(0))+(v_{1b}(1)-v_{1a}(1)) = (1-0)+(H-(H+1+\epsilon))=-\epsilon < 0.$$
Based on Theorem~\ref{thm:cycle-mon}, this implies that this allocation rule is not implementable
\end{proof}

\begin{figure}[h!]
\centering
\includegraphics[scale=.5]{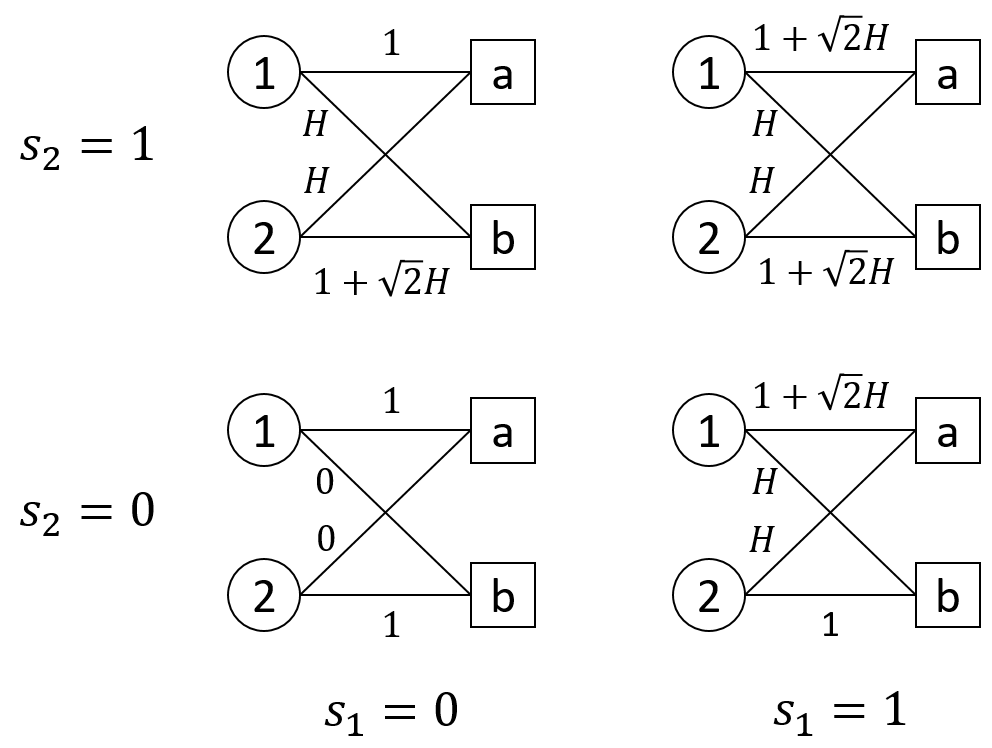}
\caption{An instance with unit-demand single-crossing valuations where no randomized truthful allocation achieves more than $\frac{\sqrt{2}+2}{4}$ of the optimal welfare.}
\centering
\label{fig:lb-random-ud}
\end{figure}

\begin{lemma}\label{lem:sclowertworan}
There exists a setting with two items and two agents with unit-demand and single crossing valuations, such that no randomized truthful mechanism achieves more than $\frac{\sqrt{2}+2}{4}$ of the optimal welfare.
\end{lemma}
\begin{proof}
Consider the setting depicted in Figure \ref{fig:lb-random-ud}, with two agents, 1 and~2, and two items, $a$ and $b$. $s_1\in\{0,1\}$ and $s_2\in \{0,1\}$.
	The values  are
	\begin{eqnarray*}
		v_{1a}(0,0)=1,\quad v_{1b}(0,0)=0,&\quad&  v_{2a}(0,0)=0,\quad  v_{1b}(0,0)=1,\\
		v_{1a}(1,0)=1+\sqrt{2}H,\quad  v_{1b}(1,0)=H,&\quad& v_{2a}(1,0)=H,\quad  v_{1b}(1,0)=1,\\
		v_{1a}(0,1)=1,\quad  v_{1b}(0,1)=H, &\quad&v_{2a}(0,1)=H,\quad  v_{2b}(0,1)=1+\sqrt{2}H,\\
		v_{1a}(1,1)=1+\sqrt{2}H,\quad  v_{1b}(1,1)=H, &\quad&v_{2a}(1,1)=H,\quad  v_{2b}(1,1)=1+\sqrt{2}H,
	\end{eqnarray*}
for an arbitrarily large $H$. One can easily verify that the valuations are single crossing.
We claim that the following equalities hold with respect to the allocation rule of the optimal randomized mechanism:
\begin{enumerate}[(a)]
	\item For every $s_1,s_2$, $x_{1a}(s_1,s_2)=x_{2b}(s_2,s_1)$ and $x_{2a}(s_1,s_2)=x_{1b}(s_2,s_1)$.
	\item For some $q\in [0,1]$, $x_{1a}(0,0)=x_{2b}(0,0)=q$ and $x_{1\emptyset}(0,0)=x_{2\emptyset}(0,0)=1-q$.
	\item For some $p\in [0,1]$, $x_{1a}(0,1)=p$ and $x_{1b}(0,1)=1-p$.
\end{enumerate}
We next prove the above equalities.
\begin{enumerate}[(a)]
	\item Consider some implementable allocation rule $\bar{x}$, and consider the allocation rule $\tilde{x}$ where $\tilde{x}_{1a}(s_1,s_2)=\bar{x}_{2b}(s_2,s_1)$ and $\tilde{x}_{2a}(s_1,s_2)=\bar{x}_{1b}(s_2,s_1)$ for every $s_1,s_2$.
Note that the valuations are symmetric; i.e., the role of item $a$ (resp. $b$) for agent~1 is the same as the role of items $b$ (resp. $a$) for agent~2.
By symmetry, $\bar{x}$ is implementable if and only if $\tilde{x}$ is implementable, and both allocation rules have the same approximation guarantee.
Clearly, an allocation rule $x$ that applies allocation rules $\bar{x}$ and $\tilde{x}$, with probability $\frac{1}{2}$ each, maintains the same approximation guarantee. Moreover, this allocation rule satisfies the desired property.
	\item The optimal mechanism gains nothing from assigning any positive probability for allocating item $b$ to agent 1 under signal profile $(0,0)$. This is because item $b$ grants no value to agent~1, and in terms of incentives, it can only incentivize agent 1 to misreport his signal at signal profile $(1,0)$. Analogously, the optimal mechanism gains nothing from assigning any positive probability for allocating item $a$ to agent 2 under signal profile $(0,0)$. By (a),  $x_{1a}(0,0)=x_{2b}(0,0)=q$ for some $q\in [0,1]$. To conclude the proof of (b), note that the only other feasible set for the agents is the empty set (otherwise, agent~1 has some probability to get item $b$ and agent~2 has some probability to get item $a$).
	\item Consider $G_1^{\mathbf{x}}$ and the cycle $C=(0,0)\rightarrow(1,0)\rightarrow (0,0)$ in $G_1^{\mathbf{x}}$. This is the only cycle that contains the node $(1,0)$ in $G_1^{\mathbf{x}}$.
	Assume $x_{1\emptyset}(1,0)>0$.
Transferring $z\in (0,1]$ probability from $x_{1\emptyset}(1,0)$ to $x_{1a}(1,0)$ decreases the weight of the edge $(0,0)\rightarrow(1,0)$ by $z$, and increases the weight of the edge $(1,0)\rightarrow(0,0)$ by $z(1+\sqrt{2} H)>z$. Therefore, its net effect on the weight of $C$ is positive.
Transferring $z\in (0,1]$ probability from $x_{1\emptyset}(1,0)$ to $x_{1b}(1,0)$ does not affect the weight of the edge $(0,0)\rightarrow(1,0)$, and increases the weight of the edge $(1,0)\rightarrow(0,0)$ by $zH$. Therefore, its net effect on the weight of $C$ is positive.
Since transferring $x_{1\emptyset}(1,0)$ to $x_{1a}(1,0)$ and $x_{1b}(1,0)$ increases welfare and does not violate cycle monotonicity, the optimal mechanism clearly assigns no probability to $x_{1\emptyset}(1,0)$.

Now assume $x_{1\{a,b\}}(1,0)>0$. By Moving this probability to $x_{1a}(1,0)$, we get the same expected welfare at $(1,0)$, and the weight of the edges in $C$ does not change. Therefore, we may also assume the mechanism does not assign positive utility to $x_{1\{a,b\}}(1,0)$.
\end{enumerate}

According to Theorem \ref{thm:cycle-mon}, in any truthful mechanism, the weight of the cycle $C$ must be non-negative . This translates to the following condition.
	\begin{eqnarray*}
		& &\left(\E_{T\sim x_1(0,0)}[v_{1T}(0,0)]-\E_{T\sim x_1(1,0)}[v_{1T}(0,0)]\right) - \left(\E_{T\sim x_1(1,0)}[v_{1T}(1,0)]-\E_{T\sim x_1(0,0)}[v_{1T}(1,0)]\right) \\
		& = &(q-p) \ +\ \left(p(1+\sqrt{2} H) + (1-p)H-q(1+\sqrt{2} H)\right) \geq 0\\
		& \Rightarrow & q \leq p\left( 1 - \frac{1}{\sqrt{2}} \right) + \frac{1}{\sqrt{2}}.
	\end{eqnarray*}
	In the optimal mechanism, $q$ will be as large as possible in order to maximize the expected welfare at signal profile $(0,0)$. Hence, we can assume $q = p\left( 1 - \frac{1}{\sqrt{2}} \right) + \frac{1}{\sqrt{2}}$. Therefore, the approximation ratio at profile $(0,0)$ is at most $q=p\left( 1 - \frac{1}{\sqrt{2}} \right) + \frac{1}{\sqrt{2}}$. At profile $(0,1)$, if item $a$ is allocated to agent~1 (which happens with probability $p$), the welfare of the mechanism is at most $2+\sqrt{2}H$, while the welfare of the optimal allocation is $2H$. As $H$ can be arbitrarily large, this approximation ratio tends to $\frac{1}{\sqrt{2}}$. Therefore, the approximation ratio at profile $(1,0)$ is at most $\frac{p}{\sqrt{2}} + (1-p) = 1-p\left(1-\frac{1}{\sqrt{2}}\right)$. The optimal mechanism would balance between the approximation ratio at $(0,0)$ and at $(1,0)$, therefore uses $p$ that solves
$$p\left( 1 - \frac{1}{\sqrt{2}} \right) + \frac{1}{\sqrt{2}}=1-p\left(1-\frac{1}{\sqrt{2}}\right).$$
Solving for $p$, we get $p = \frac{1}{2}$. This leads to an approximation ratio of at most $\frac{2+\sqrt{2}}{4}$, as promised.
\end{proof}

\section{$n-1$ Lower Bound for Deterministic Mechanisms with Single-Crossing SOS Valuations.}
\label{sec:dc-single-param}
We show that for downward-closed environments, even if valuations satisfy a single-crossing condition and are SOS, any deterministic mechanism cannot obtain a better approximation to the optimal welfare than $n-1$.

\begin{theorem}\label{thm:dc-sc-lb}
	There exists a downward-closed environment with valuations that satisfy single-crossing for which no deterministic mechanism more than a $n-1$ fraction of the optimal welfare.
\end{theorem}
\begin{proof}
	Consider a set of $n$ bidders, where $\mathcal{I}=\{1\}\cup P(\{2,\ldots, n\})$, where $P(\{2,\ldots, n\})$ is the power set of the set $\{2,\ldots, n\}$. Only agent~1 has a signal $s_1\in \{0,1\}$, and other players do not have signals. The valuations are:
	\begin{eqnarray*}
		v_1(0)=1 &\quad &v_1(1)=1+H\\
		v_i(0)=0 &\quad &v_i(1)=H\quad\quad \quad \forall i\in \{2,\ldots, n\}
	\end{eqnarray*}
	for an arbitrary large value $H\gg1$. Once can easily verify these valuations satisfy single-crossing and SOS.
	
	Any deterministic mechanism that wants to get any approximation to the social welfare must allocate to agent~1 when $s_1=0$. In addition, if a deterministic mechanism wants to get a better approximation than $n-1$ to the optimal social welfare, agent~1 cannot be allocated when $s_1=1$. Otherwise, none of the bidders in $\{2,\ldots n\}$ can get allocated because the only set in $\mathcal{I}$ that contains agent~1 is the singleton set. Therefore, if agent~1 is allocated at $s_1=1$, the achieved welfare is $1+H$, whereas the optimal welfare is $(n-1)\cdot H$ (when serving all agents in $\{2,\ldots, n\}$). For an arbitrary large $H$ This ratio approaches $n-1$.
	
	The proof follows since serving agent~1 at $s_1=0$ and not serving agent~1 at $s_1=1$ is violates monotonicity.
\end{proof}

\begin{remark}
	The $n-1$ factor is tight for single-crossing valuations. If $[n]\in\mathcal{I}$, then the mechanism can always allocate all agents. Otherwise, one can always allocate only to the highest valued agent, which is monotone because of single crossing. Since the largest feasible set is of size at most $n-1$ in this case, allocating to the highest valued agent yields an approximation ratio of $n-1$. 

\end{remark}
\section{Results for $d$-SOS} \label{sec:d-sm}
We now extend the results in Section~\ref{sec:comb-auction-single-dimension} to the case of combinatorial $d$-SOS and combinatorial $d$-strong-SOS valuations with single-dimensional signals. We first note that if we consider $d$-SOS valuations, then Equation~\eqref{eq:decomp-sm} in the decomposition becomes
\begin{eqnarray}
W^* &  \leq & \sum_{i}v_{iT_i^*}(\mathbf{s}_{-i},0_i)+\sum_{i\ :\  s_i>0}d\cdot \left(v_{iT_i^*}(\mathbf{0}_{-i}, s_i) -v_{iT_i^*}(\mathbf{0})\right) \nonumber \\ 
& \leq & \underbrace{\sum_i v_{iT_i^*}(\mathbf{s}_{-i},0_i)}_{\text{$\mathsf{OTHER}$}}+\underbrace{\sum_{\ell=1}^{k-1}\sum_{i\ :\  s_i=\ell}d\cdot v_{iT_i^*}(\mathbf{0}_{-i}, s_i)}_{\text{$\mathsf{SELF}$}},\label{eq:decomposition-d-sm}
\end{eqnarray}

We now show the extension of Theorem~\ref{thm:k-sig-sm} to $d$-SOS valuations.

\begin{theorem}	
	For every combinatorial auction with $d$-SOS valuations over single-dimensional signals, and signal space of size $k$, i.e., $s_i \in \{0, 1,\ldots, k-1\} \, \forall i$, there exists a truthful mechanism that gives $d(k+1)+2$-approximation to the optimal social welfare. \label{thm:kd_sm}
\end{theorem}
\begin{proof}
	The mechanism is identical to \texttt{$k$-HL}, but runs \texttt{(Random Threshold)} with probability $p_{RT}=\frac{(k-1)d}{d(k+1)+2}$ and \texttt{(Random Sampling)} With probability $1-p_{RT}$.
	The mechanism was already proved to be truthful in Section~\ref{sec:k-sig-sm}. 
	
	\texttt{Random Threshold} now gives a $d(k-1)$-approximation to the new \textsf{SELF} term. The proof is the same as of Lemma~\ref{lem:self_cover}, but the extra factor of $d$ comes from the fact the the new \textsf{SELF} term is $d$ times larger.
	
	\texttt{Random Sampling} gives a $2(d+1)$-approximation to the \text{OTHER} term. While this term is the same for $d$-SOS, the new factor is due to the fact that when applying Lemma~\ref{lem:rs-value} in the proof of Lemma~\ref{lem:other_cover}, we get that $\E_{A,B}[\tilde{v}_{iT}]\geq \frac{1}{2(d+1)}v_{iT}(\mathbf{s}_{-i},0_i)$ instead of the bound we get in Equation~\eqref{eq:other_cover}. 
	
	The new approximation guarantee follows from the new decomposition, the new approximation guarantees the various mechanisms get for the terms of the decomposition, and the updated probability $p_{RT}$. 
\end{proof}

We next extend Theorem~\ref{thm:k-sig-ssm}.

\begin{theorem}	\label{thm:kd-sig-ssm}
	For every combinatorial auction with  $d$-strong-SOS valuations over single-dimensional signals, and signal space of size $k$, i.e., $s_i \in \{0, 1,\ldots, k-1\} \, \forall i$, there exists a  truthful mechanism that gives $(d(d+1)\log_2{k}+2(d+1))$-approximation to the optimal social welfare.
\end{theorem}
\begin{proof}
	The mechanism is identical to mechanism \texttt{$k$-SS} from Section~\ref{sec:k-sig-ssm}, but runs \texttt{Random Bucket} with probability $p_{RB}=\frac{d\log_2{k}}{d\log_2{k}+2}$ and \texttt{(Random Sampling)} With probability $1-p_{RB}$.
	
	The \textsf{SELF} term from Equation~\eqref{eq:selfbound_strong_submod} is now bounded via the following:
	\begin{eqnarray}
	\mathsf{SELF} &=&	\sum_{\ell=1}^{k-1}\sum_{i\ :\  s_i=\ell}d\cdot v_{iT_i^*}(\mathbf{0}_{-i}, s_i)\nonumber\\
	& = & \sum_{\ell=1}^{\log_2 k}\sum_{i\ :\ 2^{\ell-1}\leq s_i < 2^\ell} d\cdot v_{iT_i^*}(\mathbf{0}_{-i}, s_i)\nonumber\\
	& \leq & \sum_{\ell=1}^{\log_2 k}\sum_{i\ :\ 2^{\ell-1}\leq s_i < 2^\ell} d(d+1)\cdot v_{iT_i^*}(\mathbf{0}_{-i}, {2^{\ell-1}}_i), \label{eq:selfbound_d_strong_submod}
	\end{eqnarray}
	where the inequality follows the definition of $d$-strong-SOS valuations.
	
	The new bound changes the guarantee of \texttt{Random Bucket} to get a $d(d+1)\log_2{k}$-approximation to the \textsf{SELF} term, where the proof is identical to that of Lemma~\ref{lem:self_cover_strongly_submod}.
	
	As stated in Theorem~\ref{thm:kd_sm}, \texttt{Random Sampling} approximates the \textsf{OTHER} term to a factor $2(d+1)$. The proof of the new bound follows the new decomposition, the updated probabilities and the new approximation guarantees of the mechanisms being run.
\end{proof}

\end{document}